\documentclass[onecolumn,12pt]{IEEEtran}
\usepackage{amsmath}
\usepackage{latexsym}
\usepackage{amssymb}
\usepackage{bm}
  \usepackage[pdftex]{graphicx}
    \usepackage{enumitem}
      \usepackage{braket}
\usepackage{color}
\usepackage{enumitem}
\newlist{Steps}{enumerate}{2}
\setlist[Steps]{label=Phase \arabic*., font=\textbf, itemindent=*}
\newtheorem{thm}    {Theorem}
\newtheorem{lem}     {Lemma}
\newtheorem{proposition}        {Proposition}
\newtheorem{define}     {Definition}
\newtheorem{example}  {Example}

\def\Tr{\mathop{\rm Tr}\nolimits}

\def\argmax{\mathop{\rm argmax}}

\def\bW{\bm{W}}

\def\bP{\bm{P}}

\def\argmax{\mathop{\rm argmax}}
\def\uni{\mathop{\rm uni}}

\def\argmin{\mathop{\rm argmin}}

\newcommand{\bF}{\mathbb{F}}

\def\Pr{{\rm Pr}}

\newcommand{\qed}{\hfill \IEEEQED}

\newcommand{\eps}{\varepsilon}

\newcommand{\sX}{\mathsf{X}}
\newcommand{\sZ}{\mathsf{Z}}

\newcommand{\sM}{\mathsf{M}}
\newcommand{\sL}{\mathsf{L}}

\def\QED{\mbox{\rule[0pt]{1.5ex}{1.5ex}}}

\def\endproof{\hspace*{\fill}~\QED\par\endtrivlist\unskip}
 \newenvironment{proofof}[1]{\vspace*{5mm} \par \noindent
         \quad{\it Proof of #1:\hspace{2mm}}}{\qed
}

\def\Label#1{\label{#1}\ [\ #1\ ]\ }
\def\Label{\label}

\begin{document}
\title{Commitment capacity of classical-quantum channels}
\author{
Masahito~Hayashi~\IEEEmembership{Fellow,~IEEE,}\thanks{Masahito Hayashi is with 
Shenzhen Institute for Quantum Science and Engineering, Southern University of Science and Technology,
Nanshan District, Shenzhen, 518055, China,
International Quantum Academy (SIQA), Futian District, Shenzhen 518048, China,
Guangdong Provincial Key Laboratory of Quantum Science and Engineering,
Southern University of Science and Technology, Nanshan District, Shenzhen 518055, China,
and
Graduate School of Mathematics, Nagoya University, Nagoya, 464-8602, Japan.
(e-mail:hayashi@sustech.edu.cn, masahito@math.nagoya-u.ac.jp)}
Naqueeb Ahmad Warsi \thanks{Naqueeb Ahmad Warsi is working as an Assistant Professor at the Indian Statistical Institute, Kolkata, 700108, India (email: naqueebwarsi@isical.ac.in)}}
\date{}
\maketitle
\begin{abstract}
We study commitment scheme for classical-quantum channels. To accomplish this we define various notions of commitment capacity for these channels and prove matching upper and lower bound on it in terms of the conditional entropy. Our achievability (lower bound) proof is quantum generalisation of the work of one of the authors (arXiv:2103.11548) which studied the problem of secure list decoding and its application to bit-string commitment. The techniques we use in the proof of converse (upper bound) is similar in spirit to the techniques introduced by Winter, Nascimento and Imai (Cryptography and Coding 2003) to prove upper bound on the commitment capacity of classical channels. However, generalisation of this technique to the quantum case is not so straightforward and requires some new constructions, which can be of independent interest. 
\end{abstract}
\section {Introduction}
Most of the modern protocols which are used to securely encrypt a message are based on the notion of commitment. Commitment with respect to this secure encryption means that one of the party (Alice) involved in the protocol is able to choose a message from a set and be committed to her choice. Her commitment to this choice of message should be in such a way that while revealing this choice of message to the other party (Bob), she should not be able to reveal something to which she didn't choose and commit. To understand this intuitively, we consider the following example:  

\begin{enumerate}
\item  Alice wants to commit a message $m$ chosen from a finite set. She does this by writing the message on a paper and then locking it inside an envelope. 

\item Alice then gives the locked envelope to Bob. At a later point of time, when Bob wants to read the message $m$, 
he asks for the key from Alice so that he can open the envelop and read the message. 
\end{enumerate}

The procedure discussed in the above example needs to satisfy the following two properties: 

\begin{enumerate}[label=\roman*]
\item {\bf{Concealing}}: After receiving the locked envelope from Alice, Bob should have no idea about what is written on the paper locked inside the envelope until Alice reveals him the key to open the envelope and read the message.
\item {\bf {Binding}}: After locking the message in the envelope, 
Alice should not be able to change it after she hands over the locked envelope to Bob. 
\end{enumerate} 
 The commitment scheme tries to solve this problem without the trusted third party (locked envelope). This problem was first introduced and studied by Blum \cite{Blum-82}. However, \cite{IJL-99} and \cite {Mayers-97} showed that if there are no computational constraints on the sender and receiver then bit commitment is not possible. 
{Cr\'{e}peau \cite{CC} pointed out that
bit commitment can be realized when a binary noisy channel is available.
That is, a noisy channel (modeled as $p_{Y\mid X}$) can help in achieving commitment scheme. 
Studying this problem from information theory perspective,
Winter et al. \cite{BC1,CCDM} }gave a probabilistic definition of commitment and {}{used} information theoretic notion for secrecy (concealing). 
Using these {}{tools,}  they {}{defined} the commitment capacity of a channel and {}{showed} 
that it is equal to $\max_{p_X}H(X \mid Y)$. 
Although their direct part is sound, {}{they wrote only the sketch of the converse part.}
{}{In addition, their converse proof contains an analysis that cannot be extended to the quantum setting,
as explained later.}
{}{The reference \cite{W-Protocols} discussed the same issue in a similar way.
Later, Yamamoto et al. \cite{BC3} studied this problem
under the problem setting of multiplex coding.
Although this paper also considers the converse part as well as the direct part, their converse part is a weaker statement than the converse part of the original problem setting, as explained later.}

{}{Further,} the papers \cite{BC1,CCDM} also introduced the commitment scheme when the parties have access to a classical-quantum channel 
{}{(cq-channel)}. 
Even though they claimed that the commitment capacity 
{}{with cq-channel}
is in terms of conditional quantum entropy, 
they didn't provide the complete proof. 
In fact, this generalization is not so straight forward. 
Also, they didn't explore the possibility that 
the parties involved have more options in terms of cheating the other party
{}{in the quantum case}. 

Recently, the reference \cite{H2021} pointed out that 
a code to achieve the commitment capacity can be constructed by using a special type of list decoding. 
{}{Originally}, list decoding was proposed by 
Elias \cite{Elias} and Wozencraft \cite{Wo} independently.
Hamming distance takes a important role in the code construction 
by \cite{H2021} similar to the preceding studies \cite{BC1,CCDM,W-Protocols,BC3}.
 
We explore all these issues in this manuscript. 
In particular, we generalize the notion of interactive protocol for implementing commitment scheme introduced in \cite{BC1,CCDM} to the case of classical-quantum channel $\bW_{X \to Y}$ (in our future discussions we will omit the subscript $X \to Y.$) 
{Towards this aim,
we define the notion of active and passive attacks. 
Using these two notions, we give two types of definitions for the commitment capacity of a classical-quantum channel and denote them as $C_a(\bW)$ and $C_p(\bW)$, respectively.} 
In this manuscript, we observe that finding $C_a(\bW)$ 
and $C_p(\bW)$ 
is a difficult problem. Therefore, we study a simpler version of the interactive protocol. In this simpler version, we restrict Alice and Bob to only use invertible operations to accomplish the commitment scheme. 
Therefore, to study this special case, we introduce $C_{a, inv}(\bW)$ and $C_{p, inv}(\bW)$ which represent the commitment capacity of a classical-quantum channel when the parties are allowed to use only invertible operations.  
Further, we also explore the case when Alice and Bob implement commitment scheme by only using non-interactive protocol over a classical-quantum channel. {We define $C_{a, non}(\bW), C_{p, non}(\bW)$
as the capacities under this setting.}

We show a relation between these notions of the commitment capacity of classical-quantum channel defined in this manuscript. In particular,  the following relationship is one of the main result of this manuscript:
  \begin{align}
  \label{main}
 C_{a, non}(\bW)=C_{p, non}(\bW)= 
  C_{a, inv}(\bW)=C_{p, inv}(\bW)= \sup_{P \in \bP({\cal X}) } H(X|Y)_{P}.
 \end{align}
We obtain \eqref{main} by first showing that $C_{p, inv}(\bW) \le \sup_{P \in \bP({\cal X}) } H(X|Y)_{P}$. 
This is the converse part and it requires \emph{construction} of some functions which helps in proving 
{the converse part} by using the Fano's inequality. 
However, as explained in {}{Subsection \ref{rem1},
the references \cite{BC1,CCDM,W-Protocols} 
have a problem in the construction of the above type of function.}
This paper concretely writes down the construction of such a function 
from a general interactive quantum protocol as Proposition \ref{Fanos}
when the protocol satisfies the invertible condition.
Since any interactive protocol in the classical setting satisfies the invertible condition, our converse proof covers the classical setting without any condition.
In addition, since the reference \cite{BC3} considered the converse part only
for non-interactive protocols, it did not discuss the above type of function.

{To show the direct part, 
we prove that $ C_{a, non}(\bW) \ge \sup_{P \in \bP({\cal X}) } H(X|Y)_{P}$ by showing the existence of a non-interactive protocol which satisfies the binding and concealing property even for a cq-channel.}
While our protocol construction is quite similar to the protocol proposed by 
 the reference \cite{H2021},
our code construction is different from that by \cite{H2021} in the following point. 
The reference \cite{H2021} considered only the classical channel, and
introduced the special class of list decoding, so called secure list decoding.
Then, the reference \cite{H2021} converts {}{secure} list decoding to a non-interactive protocol. 
In this conversion, Bob applies the list decoder and gets the list of messages
in the commitment phase. 
Bob checks whether the information revealed by Alice is contained in the list in the reveal phase.
However, in the case with cq-channel, it is not so easy to construct 
the list decoder due to the non-commutativity of the density operators.
Therefore, in this paper, we construct Alice's encoder of the commitment phase in the same way as the paper \cite{H2021}.
In our constructed non-interactive protocol,
Bob does nothing in the commitment phase.
In the reveal phase, 
he applies the projection corresponding to the information
revealed by Alice to check whether Alice is honest or not.

The rest of the manuscript is as follows. Section \ref{prep} prepares
notations and the definitions used in this manuscript. 
Section \ref{formulation} mathematically formulates the commitment scheme and gives a formal mathematical definition for it. 
We also define several notions of the commitment capacity in this section. 
Section \ref{symmetric} applies our result to the case when the channel has symmetry.
Section \ref{converse} proves the converse for commitment scheme when Alice and Bob are only using reversible operations to accomplish the commitment scheme. 
Section \ref{achievable} gives a protocol when Alice and Bob are only allowed to use non-interactive protocol for accomplishing commitment scheme. The protocol is given by a conversion from a specific type of code,
and 
Section \ref{S6} is devoted to its construction.
Section \ref{S7} makes conclusion and discusses future studies.

\section{Preparation}
\label{prep}
\subsection{Notations and Information quantities}
This paper focuses on a noisy classical-quantum (cq-) channel $\bW=\{W_x\}_{x \in {\cal X}}$
from an 
input classical system ${\cal X}$ composed of finite elements
to a quantum system ${\cal H}_Y$,
where $W_x$ is the density operator on the output quantum system ${\cal H}_Y$ with input $x
\in {\cal X}$.
Also, we define the density operator $W_{P}$ on $Y$
as $W_P:= \sum_{x\in {\cal X}}P(x)W_x$.
Then, the joint cq-state $\bW\times P$ is defined as
\begin{align}
\bW\times P= \sum_{x \in {\cal X}}P(x)|x\rangle \langle x| \otimes W_x .
\end{align}
We denote the set of probability distributions on 
${\cal X}$ and the set of density operators on ${\cal H}_Y$ 
 by $\bP({\cal X})$ and ${\cal S}({\cal H}_Y)$, respectively.

$D(\rho\|\sigma)$ is the relative entropy between two density operators $\rho$ and $\sigma$, which is defined as
\begin{align}
D(\rho\|\sigma):= \Tr \rho (\log \rho-\log \sigma).
\end{align}
The sandwitched realtive entropy $\tilde{D}_\alpha(\rho\|\sigma)$ is defined as
\begin{align}
\tilde{D}_\alpha(\rho\|\sigma):= \frac{1}{\alpha-1}\log \Tr 
(\sigma^{-\frac{\alpha-1}{2\alpha}}\rho \sigma^{-\frac{\alpha-1}{2\alpha}})^\alpha.
\end{align}
Given a state $\rho_{XY}$ on $XY$,
we consider various information quantities like mutual information and conditional entropy.
When we need to clarify the state on the quantum system 
for these quantities, we add the symbol like $[\rho]$ after the information quantity. 
For example, the entropy is defined as 
$H(XY)[\rho_{XY}]:=-\Tr \rho_{XY} \log \rho_{XY}$.
Various type of conditional entropies are defined as
\begin{align}
{H}(X|Y)[\rho_{XY}]& := H(XY)[\rho_{XY}]- H(Y)[\rho_{XY}]\\
\tilde{H}_\alpha(X|Y)[\rho_{XY}]&:= \max_{\sigma\in {\cal S}({\cal H}_Y)}
-\tilde{D}_\alpha(\rho_{XY}\|  I_X \otimes \sigma).
\end{align}
Various type of mutual informations 
are defined as
\begin{align}
I(X;Y)[\rho_{XY}]& := H(XY)[\rho_{XY}]- H(X)[\rho_{XY}]- H(Y)[\rho_{XY}]\\
\tilde{I}_\alpha(X|Y)[\rho_{XY}]&:= \max_{\sigma\in {\cal S}({\cal H}_Y)}
-\tilde{D}_\alpha(\rho_{XY}\|  \rho_X \otimes \sigma).
\end{align}

Now, we consider the case when the joint state $\rho_{XY}$ is given as
$\bW\times P$
by using a distribution $P$ on ${\cal X}$.
Under the state $\bW\times P$, we change the symbol added to various information quantities, 
$[\bW\times P]$
to $~_P$. That is, we define
\begin{align}
H(XY)_P&:=H(XY)[\bW\times P] ,\quad
\tilde{H}_\alpha(X|Y)_P:=\tilde{H}_\alpha(X|Y)[\bW\times P] \\
I(X;Y)_P&:=I(X;Y)[\bW\times P],\quad
\tilde{I}_\alpha(X|Y)_P:=\tilde{I}_\alpha(X|Y)[\bW\times P].
\end{align}
In addition, we denote the trace norm of a operator $C$ 
and  the von Neumann entropy of the density $\rho$
by $\|C\|_1$ and $S(\rho)$, respectively.

\subsection{Quantum measurements}
\label{S4}
To formulate our general adaptive method for the discrimination of cq-channels,   
we prepare a general notation for quantum measurements with state changes.
A general quantum state evolution from $A$ to $B$
is written as a 
{}{completely positive trace-preserving (cptp)} map $\mathcal{M}$ from the space 
$\mathcal{T}^A$ to the space $\mathcal{T}^B$ of trace class 
operators on $A$ and $B$, respectively.
When we make a measurement on the initial system $A$, 
we obtain the measurement outcome $K$ and the resultant state on the output system $B$.
To describe this situation, we use a set $\{\kappa_k\}_{k \in \mathcal{K}}$ 
of cp maps from the space $\mathcal{T}^A$ to the space $\mathcal{T}^B$
such that
$\sum_{k \in \mathcal{K}} \kappa_k$ is trace preserving.
In this paper, since the classical feed-forward  information is assumed to be a discrete variable,
$\mathcal{K}$ is a discrete (finite or countably infinite) set.
Since it is a decomposition of a cptp map, it is often called a \emph{cp-map valued measure}, 
and an \emph{instrument} if their sum is cptp.\footnote{For simplicity, here and in the rest of the paper, we assume the set $\mathcal{K}$ to be discrete. In fact, if the Hilbert spaces $A$, $B$, etc, on which the cp maps act are finite dimensional, then every instrument is a convex combination, i.e. a probabilistic mixture, of instruments with only finitely many non-zero elements; this carries over to instruments defined on a general measurable space $\mathcal{K}$. Thus, in the finite-dimensional case the assumption of discrete $\mathcal{K}$ is not really a restriction.}
In this case, when the initial state on $A$ is $\rho$ and 
the outcome $k$ is observed with probability 
$\Tr \kappa_k(\rho)$, where
the resultant state on $B$ is $\kappa_k(\rho)/\Tr \kappa_k(\rho)$.
A state on the composite system of the classical system $K$ and the quantum $B$ 
is written as
$\sum_{k \in {\cal K}}|k\rangle \langle k| \otimes \rho_{B|k}
$, which belongs to the vector space 
$\mathcal{T}^{KB}:=
\sum_{k \in {\cal K}}|k\rangle \langle k| \otimes \mathcal{T}^B$.
The above measurement process can be written as the following cptp 
$\mathcal{E}$ map from $\mathcal{T}^A$ to $\mathcal{T}^{KB}$.
\begin{align}
\mathcal{E}(\rho):= \sum_{k \in \mathcal{K}}|k\rangle \langle k| \otimes \kappa_{k}(\rho).\label{NACLL}
\end{align}
In the following, 
both of the above cptp map $\mathcal{E}$ and 
a cp-map valued measure are called a quantum instrument.

\section{Problem formulation}
\Label{formulation}
\subsection{General protocol description}
\Label{prob}
There are two parties Alice and Bob. Alice wants to communicate a message $M$ chosen uniformly  from the set $\{1,\cdots, 2^{nR}\}$ using $n$ uses of a noisy classical-quantum (cq-) channel 
$\bW=\{W_x\}_{x \in {\cal X}}$
from an input classical system ${\cal X}$ composed of finite elements
to a quantum system ${\cal H}_Y$.
We also assume the following condition for our cq-channel $\bW$;
\begin{description}
\item[(NR)]
Any element $x \in {\cal X}$ satisfies
\begin{align}
\min_{x \in {\cal X}}
\min_{P \in \bP({\cal X}\setminus \{x\})}
D\bigg( \sum_{x' \in {\cal X}\setminus \{x\}}P(x')W_{x'} \bigg\|W_x\bigg) 
>0 .
\end{align}
This condition is called the non-redundant condition \cite{BC1,CCDM,W-Protocols}.
\end{description}

They are also allowed to use a noiseless channel any number of times. However, this whole communication process consists of two phases:
\begin{Steps}
 \item {\bf{(Commit phase)}} Based on Alice's choice of message $m \in \{1, \cdots, 2^{nR}\}$, 
 there are $n$ rounds of a multi-round of communication from Alice to Bob and Bob to Alice. In all the discussions below one round of communication ends when first Alice communicates to Bob and then Bob communicates to Alice. 
Further, Alice has a classical memory $Z$ and Bob has a quantum memory $Y'$. 

In the first round, Alice communicates $U_1 = f_1(m, Z)$ to Bob over a noiseless channel,
and also communicates $X_1 = g_1(m, Z)$ over a classical-quantum channel, which Bob receives as quantum state $W_{X_1}$ on the  quantum system $Y_1$.
Then, Bob has the state $\rho_{U_1 Y_1}(m)$.  
After receiving the quantum system $Y_1$ and the classical information $U_1$ from Alice, 
dependently on $U_1=u_1$,
Bob applies the first quantum instrument $\{\Gamma_{v_1|u_1}^{(1)}\}_{v_1 \in {\cal V}_1}:
U_1 Y_1 \to Y'_1 V_1$.
Then, Bob has the state $\rho_{Y'_1 V_1}(m)$.  
Bob sends the outcome $V_1$ to Alice.

In the same way as the above, the $i$-th round is given as follows.
Alice communicates $U_i = f_i(m, Z, V_{i-1})$ to Bob over a noiseless channel 
by using additional classical information $V^{i-1}$.
Also, she communicates $X_i = T_i(m, Z, V_{i-1})$ over a classical-quantum channel, which Bob receives as quantum state $W_{X_i}$ on the quantum system $Y_i$.
Then, Bob has the state $\rho_{Y'_{i-1} U_i Y_i}(m)$.  
After receiving the quantum system $Y_i$ and the classical information $U_i$ from Alice, 
Dependently on $u^i,v^{i-1}$,
Bob applies the $i$-th quantum instrument 
$\{\Gamma_{v_i|u^i,v^{i-1}}^{(i)}\}_{v_i \in {\cal V}_i}:
Y'_{i-1}U_i Y_i \to Y'_i V_i$, where
$u^i=(u_1, \ldots, u_i) $ and $v^{i-1}=(v_1, \ldots, v_{i-1})$.
Then, Bob has the state $\rho_{Y'_i V_i}(m)$.  
Bob sends the outcome $V_i$ to Alice.

We denote honest Bob's behavior in Phase 1 and 
Bob's arbitrary behavior in Phase 1
by ${\cal B}$ and ${\cal B}'$, respectively.
We denote the set of Alice's honest operations in Phase 1 
by ${\cal A}_1=\{{A}_1(m)\}_m$,
where ${A}_1(m)$ is Alice's honest operations in Phase 1 with $M=m$.
After $n$-th round, Bob's state is written as
$W_{{\cal B},A_1(m),X^n=x^n}^{U^n V^n Y'_n}$ or
$W_{{\cal B}',A_1(m),X^n=x^n}^{U^n V^n Y'_n}$
dependently on Alice's and Bob's operations and $X^n=x^n$.
Similarly, we define 
$W_{{\cal B},A_1(m),Z=z}^{U^n V^n Y'_n}$ and
and $W_{{\cal B},A_1(m)}^{U^n V^n Y'_n}$
Also, we define 
$W_{{\cal B},{\cal A}_1}^{U^n V^n Y'_n}
=\sum_{m} 2^{-nR}W_{{\cal B}, A_1(m)}^{U^n V^n Y'_n}$.

It is required that the whole communication process at the end of Phase $1$ doesn't reveal anything about the message $m$ to Bob until Alice reveals him the message in the reveal phase mentioned below.  
This property of Phase $1$ is called the \emph{concealing} property and is defined as follows: 
We call Phase $1$ as $\eps$ concealing for passive Bob
if we have
\begin{align} 
\frac{1}{2}\Big{\|}
W_{{\cal B},A_1(m)}^{U^n V^n Y'_n}-W_{{\cal B},A_1(m')}^{U^n V^n Y'_n}
\Big\|_1 \leq \eps. \Label{XLP}
\end{align} 
for any message pair $(m, m')$ with $m \neq m'$.
We call Phase $1$ as $\eps$ concealing for active Bob
if we have
\begin{align} 
\frac{1}{2}\Big{\|}
W_{{\cal B}',A_1(m)}^{U^n V^n Y'_n}-W_{{\cal B}',A_1(m')}^{U^n V^n Y'_n}
\Big\|_1 \leq \eps. \Label{XLPA}
\end{align} 
for any message pair $(m, m')$ with $m \neq m'$
and Bob's arbitrary behavior ${\cal B}'$ in Phase 1.
{}{The concealing property for active Bob is a stronger condition than 
the concealing property for passive Bob.}

 \item {\bf{(Reveal phase)}} In this phase, Alice reveals her message $M=m$ and 
 her private randomness $Z$ to Bob via a noiseless channel. 
 Bob tries to answer the question ``is the message revealed by Alice correct or not?'' 
For this aim, Bob applies binary valued measurements 
${\cal T}=
\{\{T_{mz},I-T_{mz}\}\}_{mz}$ 
 on the system $U^n V^n Y'_n$, where $T_{mz}$ corresponds to the ``Accept''.
To ensure that Alice doesn't cheat in the reveal phase, i.e., she is not able to reveal some wrong message to Bob, we hope that 
Phase $2$ has {}{this} property which we call as the 
{}{\emph{binding}} property.
Further, the following condition is required; 
{}{if} both Alice and Bob don't cheat, then the measurement outcome at the end of Phase $2$ should ask to Bob to accept the message revealed by Alice. This property of the protocol is called {}{\emph{correctness}}. 

Both these properties of Phase $2$ are defined mathematically as follows.
We denote the set of Alice's honest operations in Phase 2
by ${\cal A}_2=
\{{A}_2(m)\}_m$, 
where ${A}_2(m)$
is Alice's honest operations in Phase 2 with $M=m$.
The {}{correctness} condition is given as
\begin{align}
&\Pr\left\{ m \mbox{ is accepted}\mid \mbox{Alice performs } 
{A}_1(m),{A}_2(m), \mbox{ Bob performs }{\cal B},{\cal T}
\right\} \geq 1-\delta \Label{XLP3}.
\end{align}

There are two kinds of {}{binding} property.
In {}{biding property}, we always assume that Bob is honest. 
The following is the {}{binding} property for passive Alice;
Any Alice's operation ${A}_2'$ for Phase 2 satisfies
\begin{align}
&\Pr\{m' \mbox{ is accepted}\mid \mbox{Alice performs }  
{A}_1(m),{A}_2', \mbox{ Bob performs }{\cal B} , {\cal T}
\} \leq \delta \Label{XLP4}
\end{align}
for $m'\neq m$.

The following is the {}{biding property} for active Alice;
When Alice's operations ${A}_1$ and ${A}_2$ for Phases 1 and 2 satisfy
the condition
\begin{align}
&\Pr\left\{ m \mbox{ is accepted}\mid \mbox{Alice performs } 
{A}_1,{A}_2, \mbox{ Bob performs }{\cal B},{\cal T}
\right\} \geq 1-\delta \Label{XLP5}
\end{align}
with an element $m$,
any Alice's operation ${A}_2'$ satisfies
\begin{align}
&\Pr\left\{ m' \mbox{ is accepted}\mid \mbox{Alice performs } 
{A}_1,{A}_2', \mbox{ Bob performs }{\cal B},{\cal T}
\right\} \leq \delta. \Label{XLP6}
\end{align}
{}{By combining the correctness and the binding property,}
the set of the above conditions is called
a $\delta$-binding condition.
\end{Steps}

The above protocol is written as a combination of
the four parts ${\cal A}_1,{\cal A}_2,{\cal B}$, and ${\cal T}$.
The tuple $({\cal A}_1,{\cal A}_2,{\cal B},{\cal T})$ is called
a protocol with $n$ rounds and is denoted by ${\cal P}$.
We denote the minimum value $\eps$ to satisfy the condition \eqref{XLPA} (the condition \eqref{XLP})
under the protocol ${\cal P}=({\cal A}_1,{\cal A}_2,{\cal B},{\cal T})$ by  $\eps_a({\cal P})$
($\eps_p({\cal P})$).
The value $\eps_a({\cal P})$
($\eps_p({\cal P})$)
depends only on Alice's operation ${\cal A}_1$ in Phase 1 (Alice's operation ${\cal A}_1$ and 
Bob's operation ${\cal B}$ in Phase 1). 
The value $\eps_a({\cal P})$
($\eps_p({\cal P})$) is called 
the active concealing parameter (the passive concealing parameter).
We denote the minimum value $\delta$ to satisfy the conditions \eqref{XLP3} 
and \eqref{XLP4} 
under the protocol ${\cal P}=({\cal A}_1,{\cal A}_2,{\cal B},{\cal T})$ by $\delta_p({\cal P})$,
which is called the passive binding parameter.
We denote the minimum value $\delta$ to satisfy the conditions \eqref{XLP3}, \eqref{XLP5}, 
and \eqref{XLP6} 
under the protocol ${\cal P}=({\cal A}_1,{\cal A}_2,{\cal B},{\cal T})$ by $\delta_a({\cal P})$,
which is called the active binding parameter.
Also, the value $R$ is called the rate of the protocol ${\cal P}$ and is denoted by 
$R({\cal P})$.

\subsection{Two subclass of protocols}
Since it is not so easy to discuss a general protocol,
we introduce the invertible condition for ${\cal A}_1$ as follows.
That is, we introduce the class of invertible protocols as the first subclass.
Alice's honest operation ${\cal A}_1$ in Phase 1 is called invertible
when there exist set of TP-CP maps
$\Lambda_{Y'_1V_1 \to U_1 Y_1}, \Lambda_{Y'_2V_2 \to Y'_1 U_2 Y_2 },
\ldots,\Lambda_{Y'_nV_n \to Y'_{n-1} U_n Y_n }$ such that
the relations
\begin{align}
\begin{aligned}
\Lambda_{Y'_1V_1 \to U_1 Y_1}(\rho_{Y'_1V_1}(m)) &=\rho_{U_1 Y_1}(m), \\
\Lambda_{Y'_2V_2 \to Y'_1 U_2 Y_2 }(\rho_{Y'_2V_2}(m))& =\rho_{Y'_1 U_2 Y_2}(m),\\
\vdots& \\
\Lambda_{Y'_nV_n \to Y'_{n-1} U_n Y_n }(\rho_{Y'_n V_n}(m))&
=\rho_{Y'_{n-1} U_n Y_n }(m)
\end{aligned}
\label{CCA}
\end{align}
hold for any $m$.
When all densities $W_x$ are commutative with each other,
any protocol ${\cal P}$ is invertible.

Next, as another subclass,
we introduce the class of non-interactive protocols.
When the classical communication for the variable $U_i$ nor $V_i$
is communicated in Phase 1,
the protocol ${\cal P}=({\cal A}_1,{\cal A}_2,{\cal B},{\cal T})$ by $\delta_p({\cal P})$
is called non-interactive.
Clearly, the class of non-interactive protocols is included in the class of invertible protocols .

\subsection{Asymptotic analysis}
To study the asymptotic limitation of the performance, 
we focus on the rate $R$.
The rate $R$ is called achievable with active attack (passive attack)
under the cq-channel $\bW$
when there exists a sequence of protocols $\{{\cal P}_n\}_{n=1}^{\infty}$ such that
${\cal P}_n$ is a protocol with $n$ rounds, 
$R=\lim_{n\to\infty} R({\cal P}_n)$, 
$\lim_{n\to\infty} \eps_a({\cal P}_n)=0$
($\lim_{n\to\infty} \eps_p({\cal P}_n)=0$)
and 
$\lim_{n\to\infty} \delta_a({\cal P}_n)=0 $ 
($\lim_{n\to\infty} \delta_p({\cal P}_n)=0 $).
The supremum of achievable rate under the cq-channel $\bW$
with active attack (passive attack)
is called the commitment capacity of $\bW$ with active attack (passive attack)
and is denoted by
$C_a(\bW)$ ($C_p(\bW)$).

The rate $R$ is called achievable with invertible protocols and with active attack (passive attack)
under the cq-channel $\bW$
when there exists a sequence of invertible protocols $\{{\cal P}_n\}_{n=1}^{\infty}$ such that
${\cal P}_n$ is an invertible protocol with $n$ rounds, 
$R=\lim_{n\to\infty} R({\cal P}_n)$, 
and
$\lim_{n\to\infty} \eps_a({\cal P}_n)=0$
($\lim_{n\to\infty} \eps_p({\cal P}_n)=0$)
and 
$\lim_{n\to\infty} \delta_a({\cal P}_n)=0 $ 
($\lim_{n\to\infty} \delta_p({\cal P}_n)=0 $).
The supremum of achievable rate with invertible protocols and with active attack (passive attack)
under the cq-channel $\bW$
is called the invertible commitment capacity with active attack (passive attack)
of $\bW$ and is denoted by
$C_{a, inv}(\bW)$
($C_{p, inv}(\bW)$).
In the same way, we define
an achievable rate with non-interactive protocols, and 
the non-interactive commitment capacity with active attack (passive attack)
$C_{a, non}(\bW)$
($C_{p, non}(\bW)$).

From these definitions, we have the following inequalities
\begin{alignat}{2}
  C_{a}(\bW)& \le & C_{p}(\bW) \nonumber\\
      \mid \vee      & & \mid \vee  \nonumber\\
  C_{a, inv}(\bW)& \le & C_{p, inv}(\bW) \label{MO1}\\
      \mid \vee      & & \mid \vee  \nonumber\\
  C_{a, non}(\bW)& \le & C_{p, non}(\bW) .\nonumber
\end{alignat}
When all densities $W_x$ are commutative with each other,
we have $  C_{a}(\bW)=  C_{a,inv}(\bW)$ and $  C_{p}(\bW)=  C_{p,inv}(\bW)$.

Then, we have the following theorem.
\begin{thm}
Assume Condition (NR). Then, we have the following relations;
 \begin{align}
 C_{a, non}(\bW)=C_{p, non}(\bW)= 
  C_{a, inv}(\bW)=C_{p, inv}(\bW)= \sup_{P \in \bP({\cal X}) } H(X|Y)_{P}.
 \end{align}
\hfill $\square$\end{thm}
This theorem is composed of two parts because of \eqref{MO1}.
\begin{align}
 C_{a, non}(\bW) &\ge \sup_{P \in \bP({\cal X}) } H(X|Y)_{P} \Label{MO2}\\
 C_{p, inv}(\bW) &\le \sup_{P \in \bP({\cal X}) } H(X|Y)_{P}\Label{MO3} .
\end{align}
{}{That is, separating active and passive scenarios, 
we clarify what properties are used in the direct and converse parts in the above way.}
 
 \section{Symmetric channel}\Label{symmetric}
\subsection{Formulation}
 As a typical example of cq-channel, we consider symmetric channel.
We consider a finite group ${\cal G}$ as the input classical system ${\cal X}$,
and a state $\rho$ on the quantum system ${\cal H}_Y$.
Also, we consider a unitary representation $U$ of ${\cal G}$ on ${\cal H}_Y$\cite{Hgroup}.
That is, for an element $g \in {\cal G}$, the unitary $U_g$ is defined to satisfy 
the following conditions; $U_e=I$ and $U_{g}U_{g'}=U_{gg'}$, where $e \in {\cal G}$
is the unit element.
When $U_e=I$ and 
there exists a complex number $e^{i \theta(g,g')}$ for 
$g,g' \in {\cal G}$ such that $
U_{g}U_{g'}=e^{i \theta(g,g')} U_{gg'}$,
the set of unitaries $\{U_g \}_{g \in {\cal G}}$
is called a projective unitary representation \cite{Hgroup}.
In the following, we assume that $\{U_g \}_{g \in {\cal G}}$ forms 
a projective unitary representation.

Then, we define the cq-channel as
 $W_g:= U_g \rho U_g^\dagger$.
 This channel is called a symmetric channel.
When $\{U_g \}_{g \in {\cal G}}$ forms 
a projective unitary representation,
we have
 $W_{g g'}:=U_g U_{g'} \rho U_{g'}^\dagger U_g^\dagger$.
Hence, we do not need to care the phase factor $e^{i \theta(g,g')} $
when we focus on the states $\{W_{g}\}_{g \in {\cal G}}$.

This channel with the commutative group was discussed in  
the reference \cite[Section VII-A-2]{hayashi2015quantum}.
The paper \cite{korzekwa2019encoding} studied such a channel model in the context of resource theory of asymmetry in the pure state case.
Recently, the papers \cite{hayashiWang,Wu} addressed this type of channels in the context of 
dense coding and private dense coding.
This class of cq-symmetric channels is a quantum generalization of
a regular channel \cite{delsarte1982algebraic}, which is a useful class of
channels in classical information theory.
This class of classical channels is often called
generalized additive \cite[Section V]{hayashi2011exponential}
or conditional additive \cite[Section 4]{hayashi2011exponential}
and contains a class of additive channels as a subclass.
Such a channel appears even in wireless communication
by considering binary phase-shift keying (BPSK) modulations \cite[Section 4.3]{hayashi2020finite}.
Its most simple example is the binary symmetric channel (BSC).

In the above symmetric channel, we define 
the stabilizer ${\cal K}\subset {\cal G}$ as
\begin{align}
{\cal K}:= \{g \in {\cal G}| W_g=W_e\},
\end{align}
where $e\in {\cal G} $ expresses the unit element of the group ${\cal G}$.

Since an element of ${\cal K}$ output the same state as the unit element $e$,
we consider the channel with the input system ${\cal X}:={\cal G}/{\cal K}$
as $W_{[g]}:= W_g$.
We call this type of channel an induced symmetric channel.
\begin{lem}\label{XLPLL4}
Any induced symmetric channel satisfies the Condition (NR).
\hfill $\square$\end{lem}

\begin{example}\Label{Ex1}
We consider the case when ${\cal G}=\mathbb{Z}_d$ and ${\cal H}$ is spanned by
$\{|j \rangle\}_{j=0}^{d-1}$.
We define the representation 
$U_g:=\sZ^g$,
where
$\sZ:=\sum_{j=0}^{d-1} e^{2\pi j i/d }|j\rangle \langle j| $.
We define 
$|\phi\rangle:= \sum_{j=0}^{d-1}a_j | j\rangle$ with $a_j \neq 0$ for $j=0, \ldots, d-1$.
When we choose the state $\rho$ to be $|\phi\rangle \langle \phi|$, 
the vectors $\{U_g |\phi\rangle\}_{g=0, \ldots, d-1}$ are linearly independent.
Hence, we have ${\cal K}=\{0\}$.
\hfill $\square$
\end{example}

\begin{example}\Label{Ex2}
Next, we consider the case with $d=p q$ in Example \ref{Ex1}.
We choose $|\psi\rangle := \sum_{j=0}^{p-1}  b_j |q j \rangle$
with $b_j \neq 0$ for $j=0, \ldots, p-1$.
When we choose the state $\rho$ to be $|\psi\rangle \langle \psi|$, 
the vectors $\{U_j |\psi\rangle\}_{j=0}^{p-1}$ are linearly independent
and $ U_j |\psi\rangle= U_{j+ pk} |\psi\rangle$ for $j=0, \ldots, p-1$ and $k=0, \ldots, q-1 $.
Hence, we have ${\cal K}=\{ pk \}_{k=0}^{q-1}$.
\hfill $\square$
\end{example}

\subsection{Calculation of commitment capacity}
To calculate the commitment capacity
$\sup_{P \in \bP({\cal X}) } H(X|Y)_{P}$
of the induced channel, we prepare the following lemma.

\begin{lem}\label{XLPLL}
The function $P_X \mapsto H(X|Y)_{P_X}$ is concave.
\hfill $\square$\end{lem}

In addition, 
for the calculation of the quantity 
$\sup_{P \in \bP({\cal X}) } H(X|Y)_{P}$, we prepare the following things.
A projective unitary representation $\{U_g \}_{g \in {\cal G}}$ on the Hilbert space ${\cal H}$
is called irreducible
when the following condition holds;
When a subspace ${\cal H}' $ of ${\cal H}$ satisfies
the condition $U_g {\cal H}'={\cal H}' $ for $g \in {\cal G}$,
${\cal H}' $ is ${\cal H}$ or $0$.
An example of irreducible projective representation is given in Example \ref{Ex5}.
When ${\cal G}$ is a commutative group and 
$\{U_g \}_{g \in {\cal G}}$ is an irreducible unitary representation of ${\cal G}$
on ${\cal H}$,
the dimension of ${\cal H}$ is $1$.

We denote the set of the irreducible projective unitary representations by $\hat{\cal G}$.
For an element $\lambda \in \hat{\cal G}$, we denote the corresponding representation space and the corresponding unitary representation 
by ${\cal H}_\lambda$ and $U^\lambda$, respectively.
We denote the dimension of ${\cal H}_\lambda$ by $d_\lambda$.
Generally, the representation space ${\cal H}_Y$
can be written as
\begin{align}
{\cal H}_Y=\bigoplus_{\lambda \in \hat{\cal G}} {\cal H}_{\lambda} \otimes 
\mathbb{C}^{n_\lambda},
\end{align}
where $n_\lambda$ expresses the multiplicity of the irreducible unitary presentation $U^\lambda$.

\begin{example}\Label{Ex3}
To see the multiplicity in the most simple example, 
we consider the case of commutative group ${\cal G}=\mathbb{Z}_d$ with ${\cal H}$ spanned by 
$\{|j\rangle\}_{j=1}^n$.
Assume that $U_g$ is given as
$\sum_{j=1}^n e^{i 2 \pi g/d }|j\rangle \langle j|$ for $g \in \mathbb{Z}_d$.
In this case, 
$U_{g,j}:= e^{i 2 \pi g/d }|j\rangle \langle j|$ is an irreducible representation on the one-dimensional space
spanned by $ |j\rangle$.
The representation $U_{g,j}$ has the same structure as $U_{g,0}$ for any $j =1, \ldots, n$, where
$U_{g,0}:= e^{i 2 \pi g/d }|0\rangle \langle 0|$.
Hence, the representation $U_{g,j}$ is equivalent to the representation $U_{g,0}$.
The representation $U_g$ contains $n$ irreducible representations equivalent to the representation $U_{g,0}$.
Now, we consider the space spanned by $\{ |0,j\rangle\}_{j=1}^n$, which equals 
${\cal H}_0 \otimes \mathbb{C}^n$, where 
${\cal H}_0$ is spanned by $|0\rangle$ and $\mathbb{C}^n$ is spanned by 
$\{|j\rangle\}_{j=1}^n$.
When the representation on ${\cal H}_0 \otimes \mathbb{C}^n$ is given as $U_{g,0} \otimes I $,
this representation has the same structure as the above representation $U_g$.
In this case, the dimension of the second space expresses the number of 
the same representation $n$, which is considered as the multiplicity.
In addition,$U_g$ is a constant times of the identity $I$, 
${\cal K}={\cal G}$.
\hfill $\square$
\end{example}

When the representation space ${\cal H}_Y$ does not contain
a subspace that equivalent to the irreducible representation space
${\cal H}_\lambda$, $n_\lambda$ is zero.
We denote the projection to the space ${\cal H}_{\lambda} \otimes 
\mathbb{C}^{n_\lambda}$ by $P_\lambda$.
We define the state $ \rho_\lambda$ on $ \mathbb{C}^{n_\lambda}$
and the probability $p(\lambda)$ as
\begin{align}
p(\lambda):= \Tr P_\lambda \rho,\quad
\rho_\lambda:= \frac{1}{p(\lambda)}\Tr_{{\cal H}_\lambda} P_\lambda \rho P_\lambda.
\end{align}

Since the average state 
$\sum_{g \in {\cal G}} \frac{1}{|{\cal G}|}W_g$ is commutative with $U_g$ for any $g \in {\cal G}$,
Schur's lemma \cite[Lemma 2.4]{Hgroup} guarantees that 
it has the form 
$\bigoplus_{\lambda \in \hat{\cal G}} \frac{p(\lambda)}{d_\lambda} 
I_\lambda\otimes \sigma_\lambda$.
Since $\sigma_\lambda$ coincides with $\rho_\lambda$, we have
\begin{align}
\sum_{g \in {\cal G}} \frac{1}{|{\cal G}|}W_g
=\bigoplus_{\lambda \in \hat{\cal G}} \frac{p(\lambda)}{d_\lambda} 
I_\lambda\otimes \rho_\lambda,
\end{align}
which implies that
\begin{align}
S\Big(\sum_{g \in {\cal G}} \frac{1}{|{\cal G}|}W_g\Big)
=\sum_{\lambda \in \hat{\cal G}}p(\lambda)
\Big(S(\rho_\lambda)+\log \frac{d_\lambda}{p(\lambda)}\Big).
\Label{CZJ}
\end{align}

Using \eqref{CZJ} and Lemma \ref{XLPLL}, we can show the following lemma.
\begin{lem}\Label{CXL}
For an induced symmetric channel, 
the uniform distribution achieves the commitment capacity.
\hfill $\square$
\end{lem}
In the above lemma, we need to address the induced symmetric channel
instead of the symmetric channel 
because 
the symmetric channel  does not satisfy Condition (NR) unless ${\cal K}=\{e\}$.


We denote the uniform distribution on the set ${\cal X}$ of inputs 
of the induced symmetric channel by $P_{\uni,{\cal X}}$.
Then, due to Lemma \ref{CXL}, the commitment capacity is calculated as
 \begin{align}
\sup_{P \in \bP({\cal X}) } H(X|Y)_{P}
=H(X|Y)_{P_{\uni,{\cal X}}}
=\log |{\cal X}|+ S(\rho)
-\sum_{\lambda \in \hat{\cal G}}
p(\lambda) \Big(S(\rho_\lambda)+\log \frac{d_\lambda}{p(\lambda)}\Big),\Label{XOF}
 \end{align}
where the second equation follows from \eqref{CZJ}. 
In particular, when ${\cal K}=\{e\}$, 
the symmetric channel satisfies 
 \begin{align}
\sup_{P \in \bP({\cal G}) } H(X|Y)_{P}
=\log |{\cal G}|+S(\rho)- \sum_{\lambda \in \hat{\cal G}}
p(\lambda) \Big(S(\rho_\lambda)+\log \frac{d_\lambda}{p(\lambda)}
\Big).
\Label{XOF2}
 \end{align}

In the following, we consider several typical cases.
When 
the projective representation $U$ is an irreducible representation
$U^\lambda$, we have
 \begin{align}
\sup_{P \in \bP({\cal G}) } H(X|Y)_{P}
=\log |{\cal X}|+S(\rho)- \log d_\lambda .
\Label{XOF3}
 \end{align}

When 
${\cal G}$ is a commutative group and $U$ has no multiplicity,
we have
 \begin{align}
\sup_{P \in \bP({\cal G}) } H(X|Y)_{P}
=\log |{\cal X}|+S(\rho)+ \sum_{\lambda \in \hat{\cal G}}
p(\lambda) \log p(\lambda)
\Label{XOF4}
 \end{align}
because the dimension of irreducible representation is $1$.

\begin{example}\Label{Ex4}
Eq. \eqref{XOF4} guarantees that
the capacity of the model in
 Example \ref{Ex1} is 
$ \log d+ \sum_{j=0}^{d-1} |a_j|^2\log |a_j|^2$.
Also, due to \eqref{XOF4}, 
the capacity of the model in Example \ref{Ex2}
 is calculated to
$ \log p+ \sum_{j=0}^{p-1} |b_j|^2\log |b_j|^2$.
\hfill $\square$\end{example}

\begin{example}\Label{Ex5}
We apply our result to dense coding with a general state \cite{dense,Hiroshima}.
For this aim, we consider $\mathbb{Z}_d,$, $\sZ$,  and ${\cal H}$ in the same way as Example \ref{Ex1}.
We define the operator $\sX:= \sum_{j=0}^{d-1} |j+1\rangle \langle j| $, where $|d\rangle=|0\rangle$.
For ${\cal G}=\mathbb{Z}_d^2$, we define 
$U_{j,k}:= \sX^j \sZ^k$. Since the relation
$ U_{j,k} U_{j',k'}= (e^{2\pi /d})^{j' k} U_{j+j',k+k'}$ holds,
$\{U_{j,k}\}_{(j,k)\in }\mathbb{Z}_d^2$ forms a projective irreducible representation \cite[Section 8.1.1]{Hgroup}.
The receiver has the system ${\cal H}_B$ with the same dimension as the sender's system ${\cal H}$.
Assume that the sender and the receiver share the $n$ copies of a state 
$\rho$ on the composite system ${\cal H}\otimes {\cal H}_B$.
Then, the sender is allowed to apply one of $\{U_{j,k}\}_{(j,k)\in }\mathbb{Z}_d^2$ on the system ${\cal H}$ and send it to the receiver 
as one use of the channel. 
In this situation, 
${\cal H}$ is the irreducible representation space \cite[Chapter 8]{Hgroup}, and
the space ${\cal H}_B$ shows the multiplicity.
We assume that 
$\rho$ is not commutative with $ U_{j,k}$ unless $(j,k)=(0,0)$.
Then, we find that ${\cal K}=\{(0,0)\}$.
By using \eqref{XOF2}, 
the capacity is calculated to
$\log d+S(\rho)- S(\rho_B) $.
\hfill $\square$\end{example}
 
\subsection{Proofs}
This subsection proves the lemmas stated in this section.
\subsubsection{Proof of Lemma \ref{XLPLL4}}
We show Condition (NR) by contradiction.
We define the set ${\cal X}_0 \subset {\cal X}={\cal G}/{\cal K}$ as
\begin{align}
{\cal X}_0:=\Big\{ [g] \in {\cal X}\Big|  W_g \hbox{ cannot be written as }
\sum_{ [g'] \in  {\cal X} \setminus \{[g]\}} P([g'])W_{g'}
\Big\}.
\end{align}
We assume that Condition (NR) does not hold, i.e., 
${\cal X}_0 \neq {\cal X}$.
We choose an element $[g] \in {\cal X}\setminus {\cal X}_0$ and 
a distribution $P$ on ${\cal X}\setminus \{[g]\}$ such that
$ W_g=\sum_{ [g'] \in  {\cal X}} P([g'])W_{g'}$.
For an element $[g_0]\in {\cal X}_0 $,
we have 
\begin{align}
 W_{g_0}=& U_{g_0 g^{-1}} W_gU_{g_0 g^{-1}}^\dagger
=\sum_{ [g'] \in  {\cal X}} P([g'])U_{g_0 g^{-1}} W_{g'}U_{g_0 g^{-1}}^\dagger\nonumber \\
=&\sum_{ [g'] \in  {\cal X}} P([g'])W_{g_0 g^{-1} g'}
=\sum_{ [g'] \in  {\cal X}} P([g g_0^{-1}g'])W_{ g'},
\end{align}
which implies the contradiction to the condition $[g_0]\in {\cal X}_0 $.
Hence, Condition (NR) holds.
\endproof

\subsubsection{Proof of Lemma \ref{XLPLL}}
We consider the state
$\rho=\lambda |0\rangle \langle 0|_Z \otimes 
\sum_{x \in {\cal X}} P_0(x) \otimes W_x
  +
  (1-\lambda) |1\rangle \langle 1|_Z \otimes 
\sum_{x \in {\cal X}} P_1(x) \otimes W_x$ on the system $Z,X,Y$.
Then, we have
\begin{align}
\lambda H(X|Y)_{P_0}+(1-\lambda) H(X|Y)_{P_1}
=H(X|YZ)[\rho] \le H(X|Y)[\rho]
=H(X|Y)_{\lambda P_0+(1-\lambda) P_1}.
\end{align}
\endproof

\subsubsection{Proof of Lemma \ref{CXL}}
Since a unitary operation does not change the information quantity $H(X|Y)_{P}$,
we have $H(X|Y)_{P_g}=H(X|Y)_{P}$, where
$P_g(x):= P( gx)$ for $x \in {\cal X}$.
Hence, Lemma \ref{XLPLL} implies that
\begin{align}
H(X|Y)_{P} \le H(X|Y)_{\sum_g P_g}=H(X|Y)_{P_{\uni},{\cal X}}.
\end{align}
\endproof
 
\section{Converse Part} 
\label{converse}
\subsection{Proof of \eqref{MO3}}
 The converse part \eqref{MO3} follows from the following theorem.
 \begin{thm} \Label{con}
Given an invertible protocol ${\cal P}_n=({\cal A}_1,{\cal A}_2,{\cal B},{\cal T})$,
there exists a distribution $P_X$ on ${\cal X}$ such that
\begin{align}
(1-\eps-  3\sqrt[3]{\delta}) R \leq H(X \mid Y)_{P_X}+ \frac{1 + \eta(\eps)}{n}, 
\Label{XCI}
\end{align}
where 
$R=R({\cal P}_n) $,
$\delta=\delta_p({\cal P}_n) $,
$\eps=\eps_p({\cal P}_n) $,
and
$\eta(\eps):= (\eps+1)\log(\eps+1)-\eps\log(\eps)$.
\end{thm}

To show Theorem \ref{con}, we prepare the following proposition and the following lemma.
In these statements, we consider various information quantities on the state $\rho_{2}$ after Bob's operation of Phase 2.
The state $\rho_{2}$ has the random variables $M,X^n,U^n, V^n,Z $
and the quantum system $Y'_n$.
Hence, omit the symbol added to information quantities, $[\rho_2]$.
That is, $I(M; U^n V^n Y'_n)[\rho_2]$ is simplified to $I(M; U^n V^n Y'_n)$.

\begin{proposition} 
\Label{Fanos}
We consider a protocol ${\cal P}_n=({\cal A}_1,{\cal A}_2,{\cal B},{\cal T})$,
with $R=R({\cal P}_n) $ and
$\delta=\delta_p({\cal P}_n) $,
and assume that $M$ is chosen uniformly from $[1:2^{nR}]$.
There exists a function $h(X^n U^n V^n)$ such that
the relation 
\begin{align}
\Pr \{M \neq h(X^n U^n V^n)\}
< 3\sqrt[3]{\delta}\Label{XZP}
\end{align}
holds under ${\cal A}_1,{\cal B}$.
Notice that the random variables $M,X^n, U^n, V^n$ are defined at the end of Phase 1
and they are defined with the state $\rho_2$.
\hfill $\square$
\end{proposition}

 \begin{lem}
\Label{cont}
Suppose Alice and Bob follow the protocol mentioned in section \ref{prob} such that two different messages $m \neq m'$  satisfy
\begin{equation} 
\frac{1}{2}\bigg{\|}W_{{\cal B},A_1(m)}^{U^nV^n Y'_n} - W_{{\cal B},A_1(m')}^{U^n V^n  Y'_n}\bigg{\|}_1 \leq \eps.\label{ZPR}
\end{equation}
Then, 
\begin{align}
I(M; U^n V^n Y'_n) \leq n\eps R + \eta (\eps), \label{ZOLK}
\end{align}
where 
$\eta(\eps):= (\eps+1)\log(\eps+1)- \eps\log(\eps).$
\hfill $\square$\end{lem}

\begin{proofof}{Lemma \ref{cont}}
Since 
$W_{{\cal B},{\cal A}_1}^{U^nV^n Y'_n}
=\sum_{m'=1}^{2^{nR}} \frac{1}{2^{nR}}
W_{{\cal B},A_1(m')}^{U^nV^n Y'_n}$,
the application of the triangle inequality implies
\begin{align}
\frac{1}{2}\|W_{{\cal B},A_1(m)}^{U^nV^n Y'_n} - W_{{\cal B},{\cal A}_1}^{U^nV^n Y'_n}\|_1 
\le \sum_{m'=1}^{2^{nR}}
\frac{1}{2\cdot 2^{nR}}\|
W_{{\cal B},A_1(m)}^{U^nV^n Y'_n}- W_{{\cal B},A_1(m')}^{U^nV^n Y'_n}\|_1 
\leq \eps,
\end{align}
where the second inequality follows from \eqref{ZPR}.
Then, since $$I(M; U^n V^n Y'_n) =
S(W_{{\cal B},{\cal A}_1}^{U^nV^n Y'_n})-
\sum_{m=1}^{2^{nR}} \frac{1}{2^{nR}}
S(W_{{\cal B},A_1(m)}^{U^nV^n Y'_n}),$$
\eqref{ZOLK} follows from the continuity property of 
von Neumann entropy, i.e., Fannes inequality \cite{Fannes}, \cite[Theorem 5.12]{Hbook}.
\end{proofof}

Now, we prove Theorem \ref{con} by using Proposition \ref{Fanos} and Lemmas \ref{cont} and \ref{XLPLL}. 
\begin{proofof}{Theorem \ref{con}}
Applying Fano's inequality to Proposition \ref{Fanos}, we have the relations
\begin{align}
H(M \mid X^n U^nV^n Y'_n) 
\le H(M \mid  X^n U^n V^n)
\le 1 + 3n\sqrt[3]{\delta}R. \Label{Fano2}
\end{align}
Then, we have 
\begin{align}
 &H(X^n\mid U^n V^n Y'_n) \nonumber\\
 &= H(MX^n \mid U^n V^n Y'_n) - H(M \mid X^n U^n V^n Y'_n) \nonumber\\
 & \geq  H(M \mid U^n V^n Y'_n) - H(M \mid X^n U^n V^n Y'_n) \nonumber\\
 & \overset{(a)} \geq  H(M \mid U^n V^n Y'_n) -1 - 3n\sqrt[3]{\delta}R \nonumber\\
 & =  H(M) - I(M; U^n V^n Y'_n) - 1- 3n\sqrt[3]{\delta} R\nonumber\\
 & \overset{(b)} \geq H(M) - n\eps R - \eta(\eps) -1 -  3n\sqrt[3]{\delta}R\nonumber\\
&= nR(1-\eps- 3\sqrt[3]{\delta}) - (1+\eta (\eps)),\Label{XL1}
\end{align}
where $(a)$ follows from \eqref{Fano2}, and $(b)$ follows from Lemma \ref{cont}. 

Next, we consider the following virtual operation on the state $\rho_2$.
We consider the state $\rho_3:=
\Lambda_{Y'_1V_1 \to U_1 Y_1}\circ \Lambda_{Y'_2V_2 \to Y'_1 U_2 Y_2 }\circ
\ldots\circ \Lambda_{Y'_nV_n \to Y'_{n-1} U_n Y_n }(\rho_2)$.
Under the state $\rho_3$, Bob's system is composed of quantum system $Y^n$
and the classical variable $U^n $.
Then, we have 
\begin{align}
H(X^n \mid Y^n)[\rho_3] \geq 
H(X^n \mid Y^n  U^n )[\rho_3]. \Label{XL2}
\end{align}
Also, we have
\begin{align}
& H(X^n \mid Y^n)[\rho_3] \nonumber\\
&=\sum_{i=1}^n H(X_i | X^{i-1}  Y^n)[\rho_3]\nonumber \\
& \le  \sum_{i=1}^n H(X_i  \mid Y_i)[\rho_3] \nonumber \\
&= \sum_{i=1}^n H(X| Y)_{P_{X_i}} \nonumber \\
& \le n H(X|Y)_{ \sum_{i=1}^n \frac{1}{n}P_{X_i}}.\Label{XL3}
\end{align}
where the last inequality follows from Lemma \ref{XLPLL}.
Since the operation 
$\Lambda_{Y'_1V_1 \to U_1 Y_1}\circ \Lambda_{Y'_2V_2 \to Y'_1 U_2 Y_2 }\circ
\ldots\circ \Lambda_{Y'_nV_n \to Y'_{n-1} U_n Y_n }$ is invertible,
Bob's information under the state $\rho_3$ is equivalent to 
Bob's information under the state $\rho_2$.
Hence, we have
\begin{align}
H(X^n \mid Y^n  U^n ) [\rho_3] = H(X^n\mid U^n V^n Y'_n) [\rho_2].
\Label{XL4}
\end{align}
Therefore, we have
\begin{align}
& n H(X|Y)_{ \sum_{i=1}^n \frac{1}{n}P_{X_i}} \nonumber\\
&\overset{(a)}\ge H(X^n \mid Y^n)[\rho_3] \nonumber\\
&\overset{(b)}\ge H(X^n \mid Y^n  U^n )[\rho_3] \nonumber\\
&\overset{(c)}=H(X^n\mid U^n V^n Y'_n) \nonumber\\
&\overset{(d)}\ge
nR(1-\eps- 3\sqrt[3]{\delta}) - (1+\eta (\eps)),
\end{align}
where Steps $(a)$, $(b)$, $(c)$, and $(d)$ follow from
\eqref{XL3}, \eqref{XL2}, \eqref{XL4}, and \eqref{XL1}, respectively.
Thus, we obtain \eqref{XCI}.
\end{proofof}

\subsection{Proof of Proposition \ref{Fanos}}
%
{\bf{Outline:}} In this proof, we discuss $\Pr \{M \neq h(X^n U^n V^n)\}$
when Alice behaves honestly in Phase 1 as ${\cal A}_1$ and 
Bob behaves honestly in Phase 1 as ${\cal B}$.
Hence, we omit the symbol ${\cal B}$ in the state description. We construct the required function $h(X^n U^n V^n)$ by following a sequence of steps. Here we give a brief  outline. For every $m,$ we first show the existence of a set which we call as $\mbox{Good}(m) \subseteq \mathcal{Z},$ where for every $z\in \mbox{Good}(m)$, the pair (m,z) has the property that the test $\mathcal{T} = \{T_{mz}, I- T_{mz}\}$ accepts the state $W_{{\cal A}_1(m),Z=z}^{U^n V^n Y'_n}$ with high probability. We then define functions $F(\cdot \mid m)$ and $\bar{F}(\cdot \mid m).$ Using these functions we define the function $h(\cdot).$ We then invoke the properties of the set $\mbox{Good}(m)$ and the functions $F(\cdot \mid m)$ and $\bar{F}(\cdot \mid m)$ to arrive at \eqref{XZP}. 

\noindent{\bf Step 1:}\quad 
The aim of this step is to define the set $\mbox{Good}(m)$ and deriving its properties.
Towards this, let
$f(m,z):= \Tr W_{{\cal A}_1(m),Z=z}^{U^n V^n Y'_n}T_{mz}$.
The {}{correctness} condition implies
\begin{align}
&\sum_{z} P_{Z|M=m}(z) f(m,z) \nonumber \\
=&
\Pr\left\{ m \mbox{ is accepted}\mid \mbox{Alice performs } 
{A}_1(m),{A}_2(m), \mbox{ Bob performs }{\cal B},{\cal T}
\right\} \nonumber \\
\geq &1- \delta.
\Label{comp}
\end{align}
  For each $m,$ define the set $\mbox{Good}(m)$ as follows: 
\begin{equation}
\Label{good}
\mbox{Good}(m):= \{z| f(m,z) > 1-\sqrt[3]{\delta}\}. 
\end{equation}
The existence of $\mbox{Good}(m)$ follows from the following set of inequalities: 
\begin{align}
&\Pr\left\{\mbox{Good}(m)\right\}\nonumber\\
&=\Pr\{f(Z,m) > 1-\sqrt[3]{\delta}\} \nonumber\\
&=1-\Pr\{f(Z,m) \leq 1-\sqrt[3]{\delta}\} \nonumber\\
&=1-\Pr\{1-f(Z,m) \ge \sqrt[3]{\delta}\} \nonumber\\
&\overset{(a)}\geq 1- \frac{\mathbb{E}_{Z|M=m}[ 1-f(Z,m)]}{\sqrt[3]{\delta}}  \nonumber\\
&\overset{(b)} \geq 1-\sqrt[3]{{\delta^2}},
\Label{good1}
\end{align}
where $(a)$ follows from Markov inequality and $(b)$ follows from \eqref{comp}. 
The relation \eqref{good1} guarantees that 
$\mbox{Good}(m)$ is a non-empty set.

\noindent{\bf Step 2:}\quad
The aim of Step 2 is introducing the functions $F(x^n,u^n,v^n  |mz)$, $F(x^n,u^n,v^n  |m)$
and deriving their properties.
Since the state on $U^n,V^n,Y_n'$ is determined with $X^n=x^n,Z=z$,
we use the notation, $W^{U^n=u^n,V^n=v^n,Y'_n}_{X^n=x^n,Z=z}$.
Using the operator;
\begin{align}
W^{Y'_n}_{U^n=u^n,V^n=v^n,X^n=x^n,Z=z}:=
\frac{1}{\Tr W^{U^n=u^n,V^n=v^n,Y'_n}_{X^n=x^n,Z=z}}
W^{U^n=u^n,V^n=v^n,Y'_n}_{X^n=x^n,Z=z},
\end{align}
we define the functions 
\begin{align}
F(x^n,u^n,v^n  |mz):=&
\Tr (
W^{Y'_n}_{U^n=u^n,V^n=v^n,X^n=x^n,Z=z}
\otimes | u^n,v^n\rangle \langle u^n,v^n|) T_{mz} \\
F(x^n,u^n,v^n  |m):=&
\max_{z \in \mbox{Good}(m)} 
F(x^n,u^n,v^n  |mz).\Label{XMP}
\end{align}

As shown in Step 6, we have
\begin{equation}
\mathbb{E}_{X^n U^n V^n|M=m,Z=z}{F}(X^n U^n V^n  |m'z' )
=
\Tr 
W^{U^n,V^n,Y'_n}_{{\cal A}_1(m),Z=z} T_{m'z'} 
\Label{claim2}
\end{equation}
for $m,m',z,z'$.
Then, we have
\begin{align}
\mathbb{E}_{X^n U^n V^n|M=m,Z=z}F(X^n U^n V^n  |mz)
=f(m,z).\Label{XPZ}
\end{align}
Therefore, we have
\begin{align}
& \mathbb{E}_{X^n U^n V^n|M=m}F(X^n U^n V^n  |m) \nonumber\\
&\ge
\sum_{z \in \mbox{Good}(m)}  P_{Z|M=m}(z) \mathbb{E}_{X^n U^n V^n|M=m,Z=z}F(X^n U^n V^n  |m)
\nonumber\\
&\overset{(a)}
\ge  \sum_{z \in \mbox{Good}(m)} P_{Z|M=m}(z) 
\mathbb{E}_{X^n U^n V^n|M=m,Z=z}F(X^n U^n V^n  |mz)
\nonumber\\
&\overset{(b)}
= \sum_{z \in \mbox{Good}(m)} P_{Z|M=m}(z)f(m,z) \nonumber\\
&\overset{(c)}
> \sum_{z \in \mbox{Good}(m)} P_{Z|M=m}(z)
(1-\sqrt[3]{\delta}) \nonumber\\
&\overset{(d)}
\ge(1-\sqrt[3]{{\delta^2}})(1-\sqrt[3]{\delta}) \nonumber\\
&>  1 - 2\sqrt[3]{\delta},
\Label{b1}
\end{align}
where Steps $(a)$, $(b)$, $(c)$, and $(d)$ follow from 
\eqref{XMP}, \eqref{XPZ}, \eqref{good}, and \eqref{good1}, respectively.

\noindent{\bf Step 3:}\quad
The aim of Step 3 is introducing the function $\bar{F}(x^n,u^n,v^n  |m)$, 
and deriving its property.
We define 
$\bar{F}(x^n,u^n,v^n  |m):=\max_{m'\neq m}F(x^n,u^n,v^n  |m')$.
Then, as shown below, we have
\begin{equation}
\Label{claim}
\mathbb{E}_{X^n U^n V^n|M=m}\bar{F}(X^n, U^n, V^n  |m)
  <\delta.
\end{equation}

To show \eqref{claim}, we choose 
$m'$ and $z' \in \mbox{Good}(m)$ as $\bar{F}(X^n U^n V^n  |m)=
{F}(X^n U^n V^n  |m')$ and 
${F}(X^n U^n V^n  |m')$ $={F}(X^n U^n V^n  |m'z')$.
We define ${A}_2'$ as
Alice's dishonest operation in Phase 2 to send $m'z'$ instead of $mz$.
Hence, we have
\begin{align}
&\mathbb{E}_{X^n U^n V^n|M=m}\bar{F}(X^n U^n V^n  |m)
\overset{(a)}=
\mathbb{E}_{X^n U^n V^n|M=m}{F}(X^n U^n V^n  |m')\nonumber \\
=&
\sum_{z}P_{Z|M=m}(z)
\mathbb{E}_{X^n U^n V^n|M=m,Z=z}{F}(X^n U^n V^n  |m')\nonumber \\
\overset{(b)}
=&
\sum_{z}P_{Z|M=m}(z)
\mathbb{E}_{X^n U^n V^n|M=m,Z=z}{F}(X^n U^n V^n  |m'z' )\nonumber \\
\overset{(c)}
=&
\sum_{z}P_{Z|M=m}(z)
\Tr 
W^{U^n,V^n,Y'_n}_{{\cal A}_1(m),Z=z} T_{m'z'} \nonumber \\
\overset{(d)} = &\Pr\{m' \mbox{ is accepted}\mid \mbox{Alice performs }  
{A}_1(m),{A}_2', \mbox{ Bob performs }{\cal B},{\cal T}  \} 
\overset{(e)}
\leq \delta ,
\end{align}
where Steps
$(a)$, $(b)$, $(c)$, $(d)$, and $(e)$, 
follow from the choice of $m'$, the choice of $z'$,
\eqref{claim2}, 
the definition of $A_2'$,
and \eqref{XLP4}, respectively.

\noindent{\bf Step 4:}\quad
The aim of Step 4 is introducing the function $\hat{M}(x^n, u^n, v^n)$, 
and deriving its property.
We define a function $\hat{M}(x^n, u^n, v^n)$ as follows.
\begin{align}
h(x^n, u^n, v^n):=\argmax_{m} F(x^n,u^n,v^n  |m).\Label{NCP}
\end{align}
When $m,x^n ,u^n, v^n$ satisfies the condition,
$ m \neq \hat{M}(x^n, u^n ,v^n)$,
\eqref{NCP} guarantees that
\begin{align}
{F}(x^n,u^n,v^n  |m) < \max_{m'} F(x^n,u^n,v^n  |m'),
\end{align}
which implies that
\begin{align}
\bar{F}(x^n,u^n,v^n  |m) =\max_{m'} F(x^n,u^n,v^n  |m').
\end{align}
Therefore, we have
$F(x^n,u^n,v^n  |m) \le \bar{F}(x^n,u^n,v^n  |m)$,
which implies that
\begin{align*}
1 \le \bar{F}(x^n,u^n,v^n  |m)+ (1-F(x^n,u^n,v^n  |m)).
\end{align*}
Hence, we have
\begin{align}
I[h( X^n, U^n, V^n) \neq m] \le \bar{F}(x^n,u^n,v^n  |m)+ (1-F(x^n,u^n,v^n  |m)),
\Label{XOZ}
\end{align}
where the function $I[h( X^n, U^n, V^n) \neq m] $ is defined as
\begin{align*}
I[h( X^n, U^n, V^n) \neq m]:=
\left\{
\begin{array}{ll}
1 & \hbox{ when }h( X^n, U^n, V^n) \neq m; \\
0 & \hbox{ when }h( X^n, U^n, V^n) = m .
\end{array}
\right.
\end{align*}

\noindent{\bf Step 5:}\quad
The aim of Step 5 is showing the desired relation \eqref{XZP}.
We have
\begin{align}
&\Pr\{h( X^n, U^n, V^n) \neq m\} \nonumber\\
&=\mathbb{E}_{X^n U^n V^n|M=m}
I[h( X^n, U^n, V^n) \neq m]
 \nonumber \\
&\overset{(a)} \le 
\mathbb{E}_{X^n U^n V^n|M=m}
(\bar{F}(X^n,U^n,V^n  |m)+ (1-F(X^n,U^n,V^n  |m)))
 \nonumber \\
&= 
\mathbb{E}_{X^n U^n V^n|M=m}(\bar{F}(X^n,U^n,V^n  |m))
+
\mathbb{E}_{X^n U^n V^n|M=m}(1-F(X^n,U^n,V^n  |m)))\nonumber\\
&\overset{(b)}
< \delta +2\sqrt[3]{\delta},
\end{align}
where Steps $(a)$ and $(b)$ follow from \eqref{XOZ} and
the combination of \eqref{b1} and \eqref{claim}, respectively.

Hence, we have
\begin{align}
\Pr (M \neq h(X^n U^n V^n))
= \mathbb{E}_M\Pr \{h(X^n U^n V^n) \neq m\}
< \delta +2\sqrt[3]{\delta}\le 3 \sqrt[3]{\delta},
\end{align}
which proves the desired statement \eqref{XZP}. 

\noindent
{\bf Step 6:}\quad
The claim in \eqref{claim2} can be shown as follows.
\begin{align}
&\mathbb{E}_{X^n U^n V^n|M=m,Z=z}{F}(X^n U^n V^n  |m'z' )\nonumber \\
=&
\mathbb{E}_{X^n U^n V^n|M=m,Z=z}\Tr (
W^{Y'_n}_{U^n,V^n,X^n,Z=z}
\otimes | U^n,V^n \rangle \langle U^n,V^n |) T_{m'z'} \nonumber \\
=&
\mathbb{E}_{X^n|M=m,Z=z}
\sum_{u^n,v^n}
\Tr W^{U^n=u^n,V^n=v^n,Y'_n}_{X^n,Z=z}
\Tr (
W^{Y'_n}_{u^n,v^n,X^n,Z=z}
\otimes | u^n,v^n \rangle \langle u^n,v^n |) T_{m'z'} \nonumber \\
=&
\mathbb{E}_{X^n|M=m,Z=z}
\sum_{u^n,v^n}
\Tr (
W^{U^n=u^n,V^n=v^n,Y'_n}_{X^n, Z=z}
\otimes | u^n,v^n \rangle \langle u^n,v^n |) T_{m'z'} \nonumber \\
=&
\mathbb{E}_{X^n|M=m,Z=z}
\Tr 
W^{U^n,V^n,Y'_n}_{{\cal A}_1(m),X^n,Z=z} T_{m'z'} \nonumber \\
=&
\Tr 
W^{U^n,V^n,Y'_n}_{{\cal A}_1(m),Z=z} T_{m'z'} .
\end{align}
\endproof


%

\subsection{Relation to existing converse part analyses}\label{rem1}
{}{The proof partially follows techniques similar to those used in the papers\cite{BC1,CCDM,W-Protocols}
because these studies used a statement similar to
Proposition \ref{Fanos}.
However, 
our Proposition \ref{Fanos} is different from the 
the corresponding statement in the papers \cite{BC1,CCDM,W-Protocols}.
In Proposition \ref{Fanos}, the estimate of the message $M$ is given as the function $h$ of 
$X^n U^n V^n$.
That is, the channel outputs $(Y_1, \ldots, Y_n)$ are not the input variables of our function $h$
because they do not exist in Phase 2 (Reveal phase) in our quantum setting.
In contrast, the papers \cite{BC1,CCDM,W-Protocols} used the variables 
$U^n V^n$ and $(Y_1, \ldots, Y_n)$
as the inputs of the estimate of the message $M$
because 
the channel outputs $(Y_1, \ldots, Y_n)$ exists even in Phase 2 (Reveal phase) in the classical setting.
Due to the above reason, we need to invent an estimation function $h$ different from their method.}

The paper \cite{BC3} also considered the converse part of in the classical setting
only with non-interactive protocols.
However, to derive the converse part, the paper \cite{BC3} assumes that 
Bob can recover the original message $M$ only with 
the received information via noisy channel and $Z$.
That is, the paper \cite{BC3} did not prove a statement corresponding to Proposition \ref{Fanos}.
In fact, if we show Proposition \ref{Fanos},
this method works for the converse part of non-interactive protocols, i.e, $C_{p,non}(\bW)$.
In this case, the converse part can be shown by the application of wiretap channel to 
the case when the main channel is 
 the noiseless communication from Alice to Bob 
and the wiretap channel is the channel to the output of which is accessible to Bob in Phase 1.
In addition, even when Proposition \ref{Fanos} is employed,
the simple wiretap scenario does not work in interactive setting
because the side information $V^n,U^n$ cannot be handled in 
the simple wiretap scenario.

\section{Direct part}
\label{achievable}
\subsection{Coding-theoretic formulation for non-interactive protocol}\Label{S2-1B}
To study the performance of non-interactive protocol, we formulate 
a code for a cq-channel $\bW$.
A map $\phi$ from ${\cal M}\times {\cal L}$
to ${\cal X}$ is called a encoder,
where
${\cal M}:= \{1, \ldots, \sM\}$
and 
${\cal L}:= \{1, \ldots, \sL\}$.
When Alice's message is $M \in {\cal M}$,
she selects $L \in {\cal L}$ according to the uniform distribution and 
sends $\phi(M,L)$ via a cq channel $\bW$.
Bob's verifier is 
$D=\{{\cal D}_{m,l}\}_{(m,l) \in {\cal M}\times {\cal L}
}$, where
$0 \le {\cal D}_{m,l}\le I$. 
A pair $(\phi,D)$ of an encoder and 
a verifier is called a code.

We introduce 
the parameters (A) -- (C) for an encoder $\phi$ 
and a verifier $D=\{{\cal D}_{m,l}\}_{(m,l)\in {\cal M}\times {\cal L}}$
as follows.
\begin{description}
\item[(A)] 
Verifiable condition.
\begin{align}
\eps_{A}(\phi,D)
&:=\max_{(m,l) \in {\cal M}\times {\cal L}}
\eps_{A,m,l}(\phi(m,l),D) \Label{XMP1}
 \\
\eps_{A,m,l}(x,D) 
&:=
1-\Tr [W_{x}{\cal D}_{m,l}].  
\end{align}

\item[(B)]
Concealing condition 
\begin{align}
\delta_B(\phi):= 
\max_{m,m' \in {\cal M}}
\Big\|
\sum_{l=1}^{\sL} \frac{1}{\sL}W_{\phi(m,l)}
-\sum_{l'=1}^{\sL} \frac{1}{\sL}W_{\phi(m',l')}
\Big\|_1.\Label{ZXL3}
\end{align}

\item[(C)] Binding condition.
For $x\in {\cal X}$, we define the quantity
$\delta_{C,x}(D)$ as the second largest value among 
$\{(1-\eps_{A,m,l}(x,D))\}_{(m,l) \in \sM\times \sL}$.
Then, we define
\begin{align}
\delta_{C}(D)&:=
\max_{x \in {\cal X}} \delta_{C,x}(D).
\Label{HB3E}
\end{align}
\end{description}

For a code $(\phi,D)$,  we define two numbers $|(\phi,D)|_1:= \sM$
and $|(\phi,D)|_2:= \sL$.

To construct a non-interactive protocol, we consider 
$n$ use of the cq-channel, which is written as a cq-channel
$\bW^n:=\{ W^{(n)}_{x^n}\}_{x^n \in {\cal X}^n}$, where
\begin{align}
W^{(n)}_{x^n}:= W_{x_1}\otimes \cdots \otimes W_{x_n}
\end{align}
with $x^n=(x_1, \ldots, x_n)$.
Given a code $(\phi_n,D_n)$ for the cq-channel $\bW^n$,
we construct a non-interactive protocol  with $n$ use of the channel $\bW$ as follows.
In Phase 1,
Alice chooses the random variable $Z$ as the uniform random variable $L\in {\cal L}$.
Given the message $M$, Alice chooses $X^n$ to be $\phi_n(M,L)$, and sends it to Bob via
the cq-channel $\bW^n$.
Bob receives the state $W^n_{X^n}$.
In Phase 2, Alice sends $M$ and $L$ to Bob.
Bob applies the measurement $\{{\cal D}_{M,L},I-{\cal D}_{M,L}\}$.
When Bob's outcome corresponds to ${\cal D}_{M,L}$, he accepts the message $M$.
Otherwise, he rejects it. 
This protocol accomplishes commitment instead of secrecy. 
We denote the above non-interactive protocol by ${\cal P}(\phi_n,D_n)$.
Remember that the active concealing parameter $\eps_a({\cal P})$
and the active binding parameter $\delta_a({\cal P})$
are defined for a protocol ${\cal P}$ in the end of Section \ref{prob}.
Then, we have the following lemma.
\begin{lem}\Label{CLP}
The relations
\begin{align}
\eps_a({\cal P}(\phi_n,D_n)) &= \delta_B(\phi_n) \Label{ZLP1}\\
\delta_a({\cal P}(\phi_n,D_n))&= \max (\eps_A(\phi_n,D_n) , \delta_C(D_n) \Label{ZLP2}
\end{align}
hold.
\hfill $\square$
\end{lem}

\begin{proof}
Since the quantity $\delta_B(\phi_n)$ is defined by \eqref{ZXL3},
the condition \eqref{XLPA} holds by replacing $\varepsilon$ by $\delta_B(\phi_n)$. Hence, we have \eqref{ZLP1}.

Since the quantity $\eps_A(\phi_n,D_n)$ is defined by \eqref{XMP1},
the condition \eqref{XLP5} holds by replacing $\delta$ by $\eps_A(\phi_n,D_n)$. 
Since the quantity $\delta_C(D_n)$ is defined by \eqref{HB3E},
the condition \eqref{XLP6} holds by replacing $\delta$ by $\delta_C(D_n)$. 
Hence, we have \eqref{ZLP2}.
\end{proof}

Therefore, to make a non-interactive protocol, it is sufficient to make the above type of code.

To construct a code, we introduce a pre-encoder and a pre-verifier, which are useful for this construction.
A map $\phi$ from $\tilde{{\cal M}}:=\{1, \ldots, \tilde{\sM}\}$
to ${\cal X}$ is called a pre-encoder.
We define $|\phi|:=\tilde{\sM}$.
Bob's verifier is 
$D=\{{\cal D}_{m}\}_{m \in \tilde{{\cal M}}
}$, where
$0 \le {\cal D}_{m}\le I$. 
We introduce 
the conditions (a) ,(b$\alpha$), and (c) for an pre-encoder $\phi$ 
and a pre-verifier $D=\{{\cal D}_{m}\}_{m \in \tilde{{\cal M}}}
$
as follows.
\begin{description}
\item[(a)] 
Verifiable condition.
\begin{align}
\eps_{A}({\phi},D)
&:=\max_{{m} \in \tilde{{\cal M}}}
\eps_{A,{m}}({\phi}({m}),D) 
\le  \eps_A \\
\eps_{A,{m}}(x,D) 
&:=
1-\Tr [W_{x}{\cal D}_{{m}}].  
\end{align}
\item[(b$\alpha$)]
R\'{e}nyi equivocation type of concealing condition of order $\alpha>1$.
\begin{align}
E_\alpha(\phi):=\log \tilde{\sM} - 
\min_{\sigma \in {\cal S}({\cal H}_Y)}
\frac{1}{\alpha-1}\log 
\sum_{m=1}^{\tilde{\sM}}
\frac{1}{\tilde{\sM}} 2^{(\alpha-1)\tilde{D}_\alpha(W_{\phi(m)} \|  \sigma)}.
\Label{HI3}
\end{align}

\item[(c)] Binding condition. 
For $x\in {\cal X}$, we define the quantity
$\delta_{C,x}(D)$ as the second largest value among 
$\{(1-\eps_{A,{m}}(x,C))\}_{{m} \in 
{\cal M}}$. We define
\begin{align}
\delta_{C}(D)&:=
\max_{x \in {\cal X}} \delta_{C,x}(D).
\end{align}
\end{description}

We can easily show that
\begin{align}
\tilde{H}_\alpha( \tilde{M}|Y)=E_\alpha(\phi) \Label{NNM}.
\end{align}

\subsection{Asymptotic analysis for non-interactive protocol}
Now, we show the inequality \eqref{MO2}, i.e.,
the existence of the non-interactive protocol to achieve the rate $\sup_{P_X} H(X|Y)_{P_X}$.
For this aim, we discuss a sequence of codes $ \{(\phi_n,D_n)\}$.
We say that 
a sequence of codes $ \{(\phi_n,D_n)\}$
is secure when 
$\eps_{A}(\phi_n,D_n)\to 0$,
$\delta_{B}(\phi_n)\to 0$, and
$\delta_{C}(D_n)\to 0$.
Then, we have the following theorem.

\begin{thm}\Label{TH1}
Assume Condition (NR).
For any distribution $P\in \bP({\cal X})$,
there exists a secure sequence codes $ \{(\phi_n,D_n)\}$
with 
$\sM_n:=|(\phi_n,D_n)|_1=2^{nR_1}$ and
$\sL_n:=|(\phi_n,D_n)|_2=2^{nR_2}$
when there exists a distribution $P$ on ${\cal X}$ such that
\begin{align}
R_1+R_2 <H(X)_P, ~ R_2 >I(X;Y)_P.\Label{XM1}
\end{align}
\hfill $\square$
\end{thm}

Therefore, there exists the above type of a code 
when there exists a distribution $P$ on ${\cal X}$ such that
$ R_1 < H(X|Y)_P$.
The combination of this fact and Lemma \ref{CLP} yields \eqref{MO2}.
That is, for the direct part \eqref{MO2},
it is sufficient to show Theorem \ref{TH1}.

To show Theorem \ref{TH1}, we discuss a sequence of pre-codes $ \{(\phi_n,D_n)\}$.
We say that a sequence of pre-codes $ \{(\phi_n,D_n)\}$
is $(\alpha,r_\alpha)$-secure when 
$\eps_{A}(\phi_n,D_n)\to 0$, $\delta_{C}(D_n)\to 0$,
and 
$\lim_{n\to \infty}\frac{1}{n}E_{\alpha}(\phi_n)\ge r_\alpha$ for $\alpha>1$.

\begin{thm}\Label{TH3}
Assume Condition (NR).
For any distribution $P\in {\cal P}({\cal X})$,
there exists a $(\alpha,r_\alpha)$-secure sequence 
of pre-codes $ \{(\phi_n,D_n)\}$
with 
$\tilde{\sM}_n:=|(\phi_n,D_n)|=2^{\lfloor n R_1 \rfloor+\lfloor n R_2 \rfloor}$ 
when there exists a distribution $P$ on ${\cal X}$ such that
\begin{align}
R_1+R_2 <H(X)_P, ~ r_\alpha = R_1+R_2-\tilde{I}_\alpha(X;Y)_P.
\Label{CMP}
\end{align}
\hfill $\square$
\end{thm}

\subsection{Proof of Theorem \ref{TH1}}
Here, we show Theorem \ref{TH1} by using Theorem \ref{TH3}.
Given $R_1,R_2$ that satisfies the condition \eqref{XM1},
we define 
${\cal M}_n:=\bF_2^{\lfloor n R_1 \rfloor}$
${\cal L}_n:=\bF_2^{\lfloor n R_2\rfloor }$.
Using Theorem \ref{TH3},
we choose a pre-code 
$(\tilde{\phi}_n,D_n)$, 
and
the set $\tilde{\cal M}_n$ is identified with 
$\bF_2^{ \lfloor n R_1\rfloor +\lfloor n R_2\rfloor }$.

We denote the projection from 
${\cal M}_n \oplus {\cal L}_n$
to ${\cal M}_n $ by $P$.
We randomly choose an invertible linear map $F$ from
$\tilde{\cal M}_n$ to 
${\cal M}_n \oplus {\cal L}_n$
such that $P\circ F$ satisfies the universal2 hash condition (see \cite{Carter, WC81} for more details on universal2 hash functions). 
%

Then, there exists a liner invertible function $f$ from 
$\tilde{\cal M}_n$ to 
${\cal M}_n \oplus {\cal L}_n$
such that
\begin{equation}
\|\rho_{ P\circ f (\tilde{M}),Y}
-\rho_{mix,M}\otimes \rho_E\|_1 
\le 2^{\frac{2}{\alpha}-1+\frac{\alpha-1}{\alpha}
(\log |{\cal B}|-E_\alpha(\tilde{\phi}))}\Label{NM4}
\end{equation}
for $\alpha\in (1,2]$, where
$\rho_{\tilde{M},Y}:=
\sum_{\tilde{m}  \in \tilde{\cal M}_n }
\frac{1}{|\tilde{\cal M}_n|}
|\tilde{m}\rangle \langle \tilde{m}| \otimes
W_{\tilde{\phi}_n(\tilde{m})}^{(n)}$. The inequality in \eqref{NM4}
follows because of the Proposition \ref{propouniv} mentioned below at the end of this subsection. 
We define $\phi_n(m,l):= \tilde{\phi}_n(f^{-1}(m,l))$.
We have
\begin{align}
\rho_{ P\circ f (\tilde{M}),Y}= 
\sum_{m \in {\cal M}_n}\frac{1}{|{\cal M}_n|}
|m\rangle \langle m|\otimes 
\sum_{l \in {\cal L}}\frac{1}{|{\cal L}|}W^{(n)}_{\phi_n(m,l)}.
\end{align}
Hence, 
\begin{align}
\delta_{B}(\phi_n)=
\|\rho_{ P\circ f (\tilde{M}),Y}
-\rho_{mix,M}\otimes \rho_Y\|_1 .\Label{NM3}
\end{align}
Since $r_a$ satisfies the second condition in \eqref{CMP},
when $\alpha$ is close to $1$, 
we have
\begin{align}
\lim_{n\to \infty}\frac{\log |{\cal B}|-E_\alpha(\tilde{\phi})}{n}
=
R_1 -r_\alpha= R_1-(R_1+R_2)+\tilde{I}_\alpha(X;Y)_Y
=\tilde{I}_\alpha(X;Y)_Y-R_2<0.\Label{NM2}
\end{align} 
The combination of 
\eqref{NM2},
\eqref{NM3}, and \eqref{NM4} shows that 
$\delta_{B}(\phi_n) \to 0$.
Other two conditions 
$\eps_{A}(\phi_n,D_n)\to 0$ and
$\delta_{C}(D_n)\to 0$ follow from Theorem \ref{TH3}.
\endproof

\begin{proposition}[\protect{\cite{Dupuis}\cite{H2014}}]
\label{propouniv}
Let $G$ be a universal2 hash function from 
${\cal A}$ to ${\cal B}$
Then, we have
\begin{align}
\mathbb{E}_G \|\rho_{G(A)E}
-\rho_{mix,B}\otimes \rho_E\|_1 
\le 2^{\frac{2}{\alpha}-1+\frac{\alpha-1}{\alpha}
(\log |{\cal B}|-\tilde{H}_\alpha(A|E))}
\end{align}
for $\alpha\in (1,2]$.
\hfill $\square$
\end{proposition}

\subsection{Outline of proof of Theorem \ref{TH3}}\Label{OUT-K}
Here, we present the outline of 
Theorem \ref{TH3}.
To realize Binding condition (c),
we need to exclude the existence of $x^n \in {\cal X}^n$ and $m\neq m' \in \tilde{{\cal M}}_n$
such that $1-\eps_{A,m}(x^n,D)$ and $1-\eps_{A,m'}(x^n,D)$ are far from 0.
For this aim, we 
focus on Hamming distance 
$d_H(x^n,{x^n}')$ between $x^n=(x_1^n, \ldots, x^n_n), {x^n}'=({x_1^n}', \ldots, {x^n_n}') \in {\cal X}^n$
as
\begin{align}
d_H(x^n,{x^n}'):= | \{  k|  x_k^n\neq {x_k^n}'\}|.\Label{ZL5}
\end{align}
and Hermitian matricess $\{\Xi_x\}_{x \in {\cal X}}$ to satisfy the following conditions;
\begin{align}
&\Tr [W_x \Xi_x] =0,\Label{CS1}\\
&\zeta_1:=\min_{x\neq x' \in {\cal X}} -(\Tr [W_{x'} \Xi_x]) >0, \Label{CS2}\\
&\zeta_2 :=\max_{x, x' \in {\cal X}} 
\Tr [W_{x'}(\Xi_x - \Tr [W_{x'} \Xi_x])^2]  < \infty.\Label{CS3}
\end{align}
For $x^n=(x_1^n, \ldots, x_n^n)\in {\cal X}^n$, we define
\begin{align}
\Xi_{x^n}^{(n)}:= \sum_{i=1}^n 
I^{\otimes (i-1)}\otimes 
\Xi_{x_i^n}\otimes I^{\otimes (n- i)}.
\end{align}
Then, given an encoder $\phi_n$ mapping $\tilde{\cal M}_n$ to ${\cal X}^n$,
we employ the following projection 
to Bob's decoder to include the message $m$ in his decoded list;
\begin{align}
\{ \Xi_{\phi_n(m)}^{(n)} \ge - \eps_1 n\} .\Label{CS4}
\end{align}
The projection \eqref{CS4} performs 
$1-\eps_{A,m}(x^n,D)$ small when 
$d_H(x^n, \phi_n(m))$ is larger than a certain threshold.

Indeed, we have the following lemma.
\begin{lem}\Label{LS3}
When Condition (NR) holds, there exist
functions $\{\Xi_x\}_{x \in {\cal X}}$ that satisfies the conditions \eqref{CS1},  \eqref{CS2}, and \eqref{CS3}.
\hfill $\square$\end{lem}

\begin{proof}
We show the desired statement for each $x \in {\cal X}$.
If any a self-adjoint operator $A_x$ satisfies
that 
$\Tr W_x A_x$ belongs to the convex hull of $\{ \Tr W_{x'} A_x\}_{x' \in {\cal X}\setminus\{x\}}$,
$W_x$ belongs to the set $
\{  \sum_{x' \in {\cal X} \setminus \{x\} }P(x')W_{x'}
| P \in {\cal P}({\cal X} \setminus \{x\})\}$.
Due to Condition {(NR)},
$W_x$ does not belong to the set $
\{  \sum_{x' \in {\cal X} \setminus \{x\} }P(x')W_{x'}
| P \in {\cal P}({\cal X} \setminus \{x\})\}$.
Considering the contraposition of the above statement, we have the following;
there exists a self-adjoint operator $A_x$ such that 
$\Tr W_x A_x > \Tr W_{x'} A_x  $ for $x' \in {\cal X}\setminus\{x\}$.
We choose 
a basis $\{|e_{j,x}\rangle \}_j$ to diagonal $A_x$, and define
$P_{x}(j):= \langle e_{j,x}|W_{x}|e_{j,x}\rangle$ and
$P_{x'}(j):= \langle e_{j,x}|W_{x'}|e_{j,x}\rangle$.
Then, $P_x$ does not belong to the convex hull of $\{P_{x'}\}_{x'\neq x} $.
Hence, applying Lemma 1 of \cite{H2021}, we obtained 
the desired statement for $x \in {\cal X}$.
\end{proof}

\section{Proof of Theorem \ref{TH3}}\label{S6}
\noindent {\bf Step 0}: 
We set $\overline{\sM}_n:= \frac{3}{2}\cdot 2^{ \lfloor n R_1 \rfloor+\lfloor n R_2 \rfloor}$.
Hence, $\tilde{\sM}_n= \frac{2}{3}\cdot \overline{\sM}_n$.
We prepare the verifier used in this proof as follows.
\begin{define}[Verifier $D_{\phi_n}$]\Label{Def1}
Given a distribution $P$ on ${\cal X}$, 
we define the verifier $D_{\phi_n}$ for a given encoder $\phi_n$ (a map from
$\overline{{\cal M}}_n:=
\{1, \ldots, \overline{\sM}_n\}$ to ${\cal X}^n$) in the following way.
Using the condition \eqref{CS4},
we define the projection
$\Pi_{x^n}:=
\{\Xi_{x^n}^{(n)}\ge -n\eps_1 \}$.
We define the verifier
$D_{\phi_n}=\{  \Pi_{\phi_n(m)} \}_m$.
\hfill $\square$
\end{define}

Remember that, 
for $x^n=(x^n_1, \ldots, x^n_n),{x^n}'=({x^n_1}', \ldots, {x^n_n}')\in {\cal X}^n$, 
Hamming distance $d_H(x^n,{x^n}')$ is defined to be the number of $k$ such that $x_k^n \neq {x_k^n}'$ as \eqref{ZL5} in Subsection \ref{OUT-K}.
In the proof of Theorem \ref{TH3}, 
we need to extract an encoder $\phi_n$ and elements $m \in {\cal M}_n$ that satisfies the following Hamming distance condition; \begin{align}
d_H(\phi_n(m),\phi_n(j)) >  n \eps_2 
\hbox{ for }\forall j\neq m. \Label{CC2} 
\end{align}
For this aim, given a code $\phi_n$ and a real number $\eps_2 >0$, we define the function
$\eta_{\phi_n,\eps_2}^C$
from $\overline{{\cal M}}_n $ to $\{0,1\}$ as
\begin{align}
\eta_{\phi_n,\eps_2}^C(m) &:=
\left\{
\begin{array}{ll}
0 & \hbox{ when \eqref{CC2} holds} \\
1 & \hbox{ otherwise. }
\end{array}
\right. \Label{De2}
\end{align}

As shown in Appendix \ref{S7-C}, we have the following lemma.
\begin{lem}\Label{LL12}
When a code $\tilde{\phi}_n$ defined in a subset $\tilde{{\cal M}}_n\subset 
\overline{{\cal M}}_n$
satisfies
\begin{align}
d_H(\tilde{\phi}_n(m),\tilde{\phi}_n(m'))> n \eps_2 \Label{E41}
\end{align}
for two distinct elements $ m \neq m'\in \tilde{{\cal M}}_n$,
the verifier $D_{\tilde{\phi}_n}$ defined in Definition \ref{Def1} satisfies 
\begin{align}
&\delta_{D}(D_{\tilde{\phi}_n}) 
\le
\frac{ \zeta_2}{{n} [\zeta_1\frac{\eps_2}{2}- \eps_1 ]_+^2  } .
\end{align}
\hfill $\square$
\end{lem}

\noindent {\bf Step 1}: The aim of this step is preparation of lemmas related to random coding.

\noindent To show Theorem \ref{TH3},
we assume that the variable $\Phi_n(m)$ for 
$m \in \overline{{\cal M}}_n$
is subject to the distribution $P^n$ independently.
Then, we have the following four lemmas, which are shown later.
In this proof, we treat the code $\Phi_n$ as a random variable.
Hence, the expectation and the probability for this variable
are denoted by $\mathbb{E}_{\Phi_n} $ and ${\rm Pr}_{\Phi_n}$, respectively.
{We prepare the following lemmas whose proofs are given in Appendices.}
\begin{lem}\Label{LL10}
We have the average version of Verifiable condition (a), i.e., 
\begin{align}
\lim_{n \to \infty}
\mathbb{E}_{\Phi_n} 
\sum_{m=1}^{\overline{\sM}_n}
\frac{1}{\overline{\sM}_n} 
\eps_{A,m}(\Phi_n,D_{\Phi_n}) 
=0
 \Label{ER1}.
\end{align}
\hfill $\square$
\end{lem}

\begin{lem}[\protect{\cite[Lemma 12]{H2021}}]\Label{LL11B}
When $R_1+R_2 <H(X)_P$,
for $\eps_2>0$,
we have
\begin{align}
\lim_{n \to \infty}
\mathbb{E}_{\Phi_n} 
\sum_{m=1}^{\overline{\sM}_n}
\frac{1}{\overline{\sM}_n} 
\eta_{\Phi_n,\eps_2}^C(m) =0
\Label{ER3}.
\end{align}
\hfill $\square$
\end{lem}

\begin{lem}\Label{LL19}
We choose $\sigma_{P,\alpha} \in {\cal S}({\cal H}_Y)$ as
\begin{align}
\sigma_{P,\alpha}:=  \argmin_{\sigma \in {\cal S}({\cal H}_Y)} \tilde{D}_\alpha( \bW\times P\|  \sigma \otimes P).
\Label{LL19E}
\end{align}
We have
\begin{align}
\mathbb{E}_{\Phi_n} 
\sum_{i=1}^{\overline{\sM}_n} \frac{1}{\overline{\sM}_n} 2^{(\alpha-1)\tilde{D}_\alpha(W_{\Phi_n(i)}\|\sigma_{P,\alpha}^n)}
=2^{n (\alpha-1) \tilde{I}_\alpha(X;Y)_P}
\Label{HR7}.
\end{align}
\hfill $\square$
\end{lem}

\noindent {\bf Step 2}: 
The aim of this step is the extraction of an encoder $\phi_n$ and messages $m$ 
with a small decoding error probability 
that satisfies the condition \eqref{CC2}.

We define $\eps_{3,n}$ as
\begin{align}
 \eps_{3,n}:= 
9 \mathbb{E}_{\Phi_n}\sum_{m=1}^{\overline{\sM}_n}
\frac{1}{\overline{\sM}_n} 
\Big(
\Big(\eps_{A,m}(\phi_n,D_{\Phi_n}) 
+\eta_{\Phi_n,\eps_2}^C(m) 
\Big)
\Big).
\end{align}
Here the function $\eta_{\Phi_n,\eps_2}^C$ reflects 
the Hamming distance condition \eqref{CC2}.
Lemmas \ref{LL10} and \ref{LL11B} 
guarantees that $\eps_{3,n} \to 0$.
Then, there exists a sequence of codes $\phi_n$ such that
\begin{align}
\sum_{m=1}^{\overline{\sM}_n}
\frac{1}{\overline{\sM}_n} 
\Big(\eps_{A,m}(\phi_n,D_{\phi_n}) 
+\eta_{\phi_n,\eps_2}^C(m) 
\Big)
& \le \frac{\eps_{3,n}}{3} \Label{BGC} \\
\sum_{m=1}^{\overline{\sM}_n} \frac{1}{\overline{\sM}_n} 
2^{(\alpha-1)\tilde{D}_\alpha(W_{\phi_n(m)}\|\sigma_{P,\alpha}^n)}
& \le 3\cdot 2^{n (\alpha-1) \tilde{I}_\alpha(X;Y)_P}
\Label{HR7F}.
\end{align}
Due to Eq. \eqref{BGC}, Markov inequality guarantees that
there exist $\frac{2}{3} \cdot \overline{\sM}_n$ elements 
$\tilde{{\cal M}}_n:=
\{m_1, \ldots, m_{\frac{2}{3} \cdot  \overline{\sM}_n}\}$
such that every element $m \in \tilde{{\cal M}}_n$ satisfies
\begin{align}
\eps_{A,m}(\phi_n,D_{\phi_n}) 
+\eta_{\phi_n,\eps_2}^C(m) 
\le \eps_{3,n},
\end{align}
which implies that
\begin{align}
\eps_{A,m}(\phi_n,D_{\phi_n}) &\le \eps_{3,n}  \Label{NAC1}\\
\eta_{\phi_n,\eps_2}^C(m)& =0 \Label{NAC}
\end{align}
because $\eta_{\phi_n,\eps_2}^C$ takes value 0 or 1.
Then, we define a code $\tilde{\phi}_n$
on $\tilde{{\cal M}}_n$
as $\tilde{\phi}_n(m):= {\phi}_n(m)$ for $m \in \tilde{{\cal M}}_n $.
Eq. \eqref{NAC1} guarantees Verifiable condition (a).
For $m,m'$, 
Eq. \eqref{HR7F} guarantees that
\begin{align}
\sum_{m\in \tilde{\cal M}_n} \frac{1}{|\tilde{\cal M}_n|} 
2^{(\alpha-1)\tilde{D}_\alpha(W_{\tilde{\phi}_n(m)}\|\sigma_{P,\alpha}^{\otimes n})}
& =
\sum_{m\in \tilde{\cal M}_n} 
\frac{3}{2\overline{\sM}_n} 2^{(\alpha-1)\tilde{D}_\alpha(W_{\phi_n(m)}\|\sigma_{P,\alpha}^{\otimes n})}
\le \frac{9}{2}\cdot 2^{n (\alpha-1) \tilde{I}_\alpha(X;Y)_P}\Label{ZSO}.
\end{align}

\noindent {\bf Step 3}: 
The aim of this step is the evaluation of the parameter $ \delta_{C}(D_{\tilde{\phi}_n,3}) $.

\noindent 
The relation \eqref{NAC} guarantees the condition
\begin{align}
d_H(\tilde{\phi}_n(m),\tilde{\phi}_n(m')) > n \eps_2
\Label{MUFV}
\end{align}
for $ m \neq m'\in \tilde{{\cal M}}_n$.
Therefore, Lemma \ref{LL12} guarantees 
Binding condition (c), i.e., 
\begin{align}
&
 \delta_{C}(D_{\tilde{\phi}_n}) 
\le
\frac{ \zeta_2}{{n} [\zeta_1\frac{\eps_2}{2}- \eps_1 ]_+^2  } \to 0.
\Label{EF1}
\end{align}

\noindent {\bf Step 4}: The aim of this step is the evaluation of the parameter $E_\alpha( \tilde{\phi}_n)$.

Eq. \eqref{ZSO} guarantees that
\begin{align}
&  \min_{\sigma_n \in {\cal S}({\cal H}_Y^{\otimes n})} \sum_{m\in \tilde{\cal M}_n} 
\frac{1}{|\tilde{{\cal M}}_n|}
2^{(\alpha-1)\tilde{D}_\alpha(W_{\tilde{\phi}_n(m)}\|\sigma_n)}\nonumber \\
\le & \sum_{m\in \tilde{\cal M}_n} 
\frac{1}{|\tilde{{\cal M}}_n|}
2^{(\alpha-1)\tilde{D}_\alpha(W_{\tilde{\phi}_n(m)}\|\sigma_{P,\alpha}^{\otimes n})}\nonumber \\
\stackrel{(a)}{\le}&
 \frac{9}{2}\cdot 2^{n (\alpha-1) \tilde{I}_\alpha(X;Y)_P}, 
\end{align}
where $(a)$ follows from \eqref{ZSO}.
Hence, we obtain Condition (b$\alpha$), i.e., 
the relation $\lim_{n\to \infty}\frac{1}{n}E_{\alpha}(\phi_n)\ge r_\alpha$ with \eqref{CMP} as
\begin{align}
 \lim_{n \to \infty} 
\frac{1}{n}E_\alpha( \tilde{\phi}_n) 
\ge  R_1+R_2- \tilde{I}_\alpha(X;Y)_P .
\end{align}
\endproof

\section{Conclusion}\Label{S7}
We have calculated various types of commitment capacities.
To show the direct part, we have extended the method by \cite{H2021}
to the quantum setting.
To show the converse part, we have shown Proposition \ref{Fanos},
which constructs a function to estimate the message from 
the random variables $X^n,U^n,V^n$.
This function has been constructed from an invertible protocol,
and satisfies the required property \eqref{XZP}
due to the security parameters of the original invertible protocol.
This part was omitted in the preceding papers \cite{BC1,CCDM,W-Protocols}.
Since any interactive protocol in the classical setting satisfies the invertible condition, our converse proof covers the classical setting without any condition.

However, we could not prove the converse part for a general interactive protocol in the cq-channel setting.
{When the invertible condition does not hold,} there exists no inverse TP-CP map
$\Lambda_{Y'_1V_1 \to U_1 Y_1}$, 
$\Lambda_{Y'_2V_2 \to Y'_1 U_2 Y_2 }$,
$\ldots,\Lambda_{Y'_nV_n \to Y'_{n-1} U_n Y_n }$ 
to satisfy the condition \eqref{CCA}.
Hence, the relation \eqref{XL4} does not hold in general.
We need to find another method to avoid this problem
{for a general interactive protocol}.
Therefore, it is a interesting future problem to calculate 
the capacities $C_p(\bW) $ and $ C_a(\bW)$.

In the direct part, 
we have constructed a specific code to satisfy Conditions (A), (B), 
and (C), and have converted it to a non-interactive protocol to achieve the  commitment capacity.
For this construction, we have constructed a pre-code to satisfy Conditions (a), (b$\alpha$), and (c) by using Hamming distance as Theorem \ref{TH3}.
However, we have not constructed a special type of list decoding 
unlike the reference \cite{H2021} due to the following reason.
If we apply the same list decoder, we need to apply a measurement, which 
might destroy the received quantum state.
Therefore, we can expect that this approach does not work well for cq-channels.
It is another interesting future direction to construct secure list decoding 
for a cq-channel
that has a similar performance as that in the reference \cite{H2021}.

\section*{Acknowledgments}
MH was supported in part by the National Natural Science Foundation of China (Grant No. 62171212) and
Guangdong Provincial Key Laboratory (Grant No. 2019B121203002).

\appendices

\section{Proof of Lemma \ref{LL12}}\Label{S7-C}
\noindent {\bf Step 1}: 
The aim of this step is the evaluation of $W^n_{x^n} (\Pi_{{x^n}',3})$.

\noindent 
The conditions \eqref{CS1} and \eqref{CS2} imply that 
\begin{align}
\Tr [ W_{{x^n}'}^{(n)} 
\Xi_{x^n}]
& \le -\zeta_1 d(x^n,{x^n}').
\end{align}
By using the method by \cite{H2002},
the condition \eqref{CS3} implies that 
\begin{align}
\Tr [ W_{{x^n}'}^{(n)} 
(\Xi_{x^n}^{(n)}-\Tr [ W_{{x^n}'}^{(n)} \Xi_{x^n}^{(n)}])^2]
& \le n \zeta_2 .
\end{align}
Hence, applying Chebyshev inequality to 
the variable $\xi_{x^n} (Y^n) $, 
we have
\begin{align}
W^n_{{x^n}'} (\Pi_{{x^n,2}})
=&
\Tr[ W^n_{{x^n}'} 
\{\Xi_{x^n}^{(n)}\ge -n\eps_1 \}]
\nonumber \\
\le &
\frac{n \zeta_2}{[
\zeta_1 d(x^n,{x^n}')-n \eps_1 ]_+^2  }\Label{Che}.
\end{align}

\noindent {\bf Step 2}: The aim of this step is the evaluation of smaller value of 
$\Tr[ W^n_{{x^n}} \Pi_{{\tilde{\phi}_n(m),3}}]$
and 
$\Tr [W^n_{{x^n}} \Pi_{{\tilde{\phi}_n(m'),3}}]$.
Since Eq. \eqref{E41} implies
\begin{align}
n \eps_2< d(\tilde{\phi}_n(m),\tilde{\phi}_n(m')) \le 
d_H(x^n,\tilde{\phi}_n(m)) + d_H(x^n,\tilde{\phi}_n(m')) ,
\end{align}
we have
\begin{align}
\max ([\zeta_1 d_H(x^n,\tilde{\phi}_n(m)) - n \eps_1 ]_+,
[\zeta_1 d_H(x^n,\tilde{\phi}_n(m')) - n \eps_1 ]_+)
\ge [n(\zeta_1\frac{\eps_2}{2}- \eps_1) ]_+^2 .
\end{align}
Hence, \eqref{Che} guarantees that
\begin{align}
&\min(
\Tr [W^n_{{x^n}} \Pi_{{\tilde{\phi}_n(m),3}} ],
\Tr [W^n_{{x^n}} \Pi_{{\tilde{\phi}_n(m'),3}}]
) \nonumber \\
\le &
\frac{n \zeta_2}{
\max ([\zeta_1 d(x^n,\tilde{\phi}_n(m)) - n \eps_1 ]_+^2,
[\zeta_1 d(x^n,\tilde{\phi}_n(m')) - n \eps_1 ]_+^2)} \nonumber \\
\le &
\frac{ n \zeta_2}{ [n(\zeta_1\frac{\eps_2}{2}- \eps_1) ]_+^2  } 
=\frac{ \zeta_2}{{n} [\zeta_1\frac{\eps_2}{2}- \eps_1 ]_+^2  } ,
\end{align}
which implies the desired statement.
\endproof

\section{Proof of Lemma \ref{LL10}}
To evaluate the value
\begin{align}
\mathbb{E}_{X^n} \Tr W_{X^n}^{(n)}(I- \Pi_{X^n}) 
=\mathbb{E}_{X^n} \Tr W_{X^n}^{(n)}\{\Xi_{x^n}^{(n)}< -n\eps_1 \},\Label{NAS}
\end{align}
we denote the eigenvalue of $\Xi_x$ with eigenvector 
$|e_{j,x}\rangle$ by $\xi(j,x)$.
We define the random variable $J,X$
whose joint distribution is 
$P_{JX}(jx)=P(x) \langle e_{j,x}|W_x|e_{j,x}\rangle$.
We consider their $n$ independent variables
$J^n=(J_1, \ldots, J_n)$
and $X^n=(X_1, \ldots, X_n)$. 
Hence, we define 
$ \xi^n(J^n,X^n):=\sum_{i=1}^n \xi(J_i,X_i)$.
The value \eqref{NAS}
equals the probability 
${\rm Pr }( \xi^n(J^n,X^n)  \le - n \eps_1 )$.
Since the expectation of $\xi(J_i,X_i)$ is zero
and 
the variance of $\xi(J_i,X_i)$ is upper bounded by $\zeta_2$,
this value goes to zero.
Hence, we obtain Lemma \ref{LL10}.

\section{Proof of Lemma \ref{LL19}}\Label{S7-E9}
Eq. \eqref{HR7} can be shown as follows.
\begin{align}
& \mathbb{E}_{\Phi} 
\sum_{i=1}^{\overline{\sM}_n} \frac{1}{\overline{\sM}_n} 
2^{(\alpha-1)\tilde{D}_\alpha(W_{\Phi_n(i)}\|\sigma_{P,\alpha}^
{\otimes n})}
=
\mathbb{E}_{\Phi} 
\sum_{i=1}^{\overline{\sM}_n} \frac{1}{\overline{\sM}_n} 
\prod_{j=1}^n
2^{(\alpha-1)\tilde{D}_\alpha(W_{\Phi_n(i)_j}\|\sigma_{P,\alpha})}
 \nonumber \\
=&
\sum_{i=1}^{\overline{\sM}_n} \frac{1}{\overline{\sM}_n} 
\prod_{j=1}^n
\sum_{x\in {\cal X}}P(x)
2^{(\alpha-1)\tilde{D}_\alpha(W_{x}\|\sigma_{P,\alpha})}
\nonumber \\
=&
\sum_{i=1}^{\overline{\sM}_n} \frac{1}{\overline{\sM}_n} 
\prod_{j=1}^n
2^{(\alpha-1) \tilde{D}_\alpha(\bW\times P\|  
\sigma_{P,\alpha} \otimes P)}
\nonumber \\
\stackrel{(a)}{=}&
\sum_{i=1}^{\sM_n} \frac{1}{\overline{\sM}_n} 
\prod_{j=1}^n
2^{(\alpha-1) \tilde{I}_\alpha(X;Y)_P}
=
\sum_{i=1}^{\overline{\sM}_n} \frac{1}{\overline{\sM}_n} 
2^{n (\alpha-1) \tilde{I}_\alpha(X;Y)_P}
=
2^{n (\alpha-1) \tilde{I}_\alpha(X;Y)_P}
\Label{HR7LL},
\end{align}
where $(a)$ follows from \eqref{LL19E}.
\endproof

\bibliographystyle{IEEE}

\end{document}